\documentclass{jpc}


\usepackage{amsmath,amsthm,amsfonts, enumitem, float, tabularx}
\usepackage[linesnumbered,ruled,vlined]{algorithm2e}
\usepackage{graphicx}
\usepackage{xcolor}
\usepackage{tikz-cd}
\usepackage{epsf}
\usepackage{hyperref}
\usepackage{balance}

 \newtheorem{theorem}{Theorem}
 \numberwithin{theorem}{section}
 \newtheorem{lemma}[theorem]{Lemma}
 \newtheorem{corollary}[theorem]{Corollary}
 \newtheorem{definition}[theorem]{Definition}
 
 \newtheorem{example}[theorem]{Example}

\newcommand{\View}{\texttt{View}}

\newcommand{\abs}[1]{\lvert{#1}\rvert}

\newcommand{\eps}{\epsilon}
\newcommand{\comp}{{\text{Comp}}}
\newcommand{\concomp}{{\text{ConComp}}}

\newcommand{\CF}{\mathcal{F}\text{-}FiltCon(IM)}
\newcommand{\CO}{\mathcal{G}\text{-}OdomCon(IM)}
\newcommand{\NIF}{\mathcal{F}\text{-}Filt(NIM)}
\newcommand{\NIO}{\mathcal{G}\text{-}Odom(NIM)}

\newcommand{\zo}{\{0,1\}}
\newcommand{\cM}{\mathcal{M}}
\newcommand{\cP}{\mathcal{P}}
\newcommand{\cA}{\mathcal{A}}
\newcommand{\calM}{\mathcal{M}}
\newcommand{\calA}{\mathcal{A}}
\newcommand{\calF}{\mathcal{F}}
\newcommand{\calD}{\mathcal{D}}
\newcommand{\calG}{\mathcal{G}}
\newcommand{\calO}{\mathcal{O}}

\begin{document}

\title[Concurrent Composition for Interactive DP with
Adaptive Privacy-Loss Parameters]{Concurrent Composition for Interactive Differential Privacy with
Adaptive Privacy-Loss Parameters }
\titlecomment{{\lsuper*} This article is for the TPDP special issue. A preliminary version of this paper has appeared at the ACM Conference on Computer and Communications Security (CCS) 2023.}

\author[S. Haney]{Samuel Haney}
\address{%
Tumult Labs, Raleigh, USA}

\author[M. Shoemate]{Michael Shoemate}
\address{Harvard University, Cambridge, USA}

\author[G. Tian]{Grace Tian}
\address{Harvard University, Cambridge, USA}

\author[S. Vadhan]{Salil Vadhan}
\address{Harvard University, Cambridge, USA}

\author[A. Vyrros]{Andrew Vyrros}
\address{Harvard University, Cambridge, USA}

\author[V. Xu]{Vicki Xu}
\address{Harvard University, Cambridge, USA}

\author[W. Zhang]{Wanrong Zhang}
\address{Harvard University, Cambridge, USA}


\begin{abstract}
  In this paper, we study the concurrent composition of interactive mechanisms with adaptively chosen privacy-loss parameters. In this
setting, the adversary can interleave queries to existing interactive mechanisms, as well as create new ones. We prove that  every valid privacy filter and odometer for noninteractive mechanisms
extends to the concurrent composition of interactive mechanisms if privacy loss is measured using $(\epsilon, \delta)$-DP, $f$-DP, or R\'enyi DP of fixed order. Our results offer strong theoretical foundations for enabling full adaptivity in composing differentially private interactive mechanisms, showing that concurrency does not affect the privacy guarantees. We also provide an implementation for users to deploy in practice.
\end{abstract}





\maketitle

\section{Introduction}


\subsection{Differential Privacy}

Differential privacy is a framework for protecting the privacy of individuals when analyzing data. 
It is a mathematical definition of privacy that ensures that the results of an analysis do not reveal too much information about any individual in the dataset. Because of its powerful worst-case guarantee, differential privacy has become a leading approach in privacy-preserving data analysis, where it is used to enable the analysis of sensitive data while preserving the privacy of individuals.

Differential privacy can be defined in terms of a general database space $\mathcal X$ and a binary neighboring relation on $\mathcal X$. For example, if databases contain an ordered and known number $n$ of real-valued entries, then $\mathcal X = \mathbb R^n$.  The binary relation on $\mathcal X$ specifies which datasets are \emph{neighboring}, meaning that they differ on one individual's data. For example, if $x \in \mathbb R^n$, then $x, x' \in \mathcal X$ are neighboring if they differ on one coordinate. Differential privacy requires that the output distributions should be roughly the same on the two neighboring datasets. 

\begin{definition}[Differential Privacy]\label{def.dp}
    A randomized mechanism $\mathcal M: \mathcal X \to \mathcal R$ is {\em $(\eps, \delta)$-differentially private} if for every pair of neighboring datasets $x, x' \in \mathcal X$, and for every subset of possible outputs $\mathcal S \subseteq \mathcal R$, $$Pr[\mathcal M(x) \in \mathcal S] \leq \exp(\eps) \cdot Pr[\mathcal M(x') \in \mathcal S] + \delta.$$
\end{definition}

By requiring that an analysis be robust to changes in neighboring datasets, differential privacy provides a guarantee that the privacy of individuals in the dataset is protected, regardless of what other data might be included or excluded.

In recent years, other forms of differential privacy have enjoyed use to address various shortcomings of $(\eps, \delta)$-DP (also known as approximate-DP) regarding composition. Some standard variants include R\'enyi DP (RDP) \cite{mironov2017renyi}, $f$-DP \cite{dong2019gaussian} (the formal definition is given in Section \ref{sec.fDP}), and zero-concentrated differential privacy (zCDP) \cite{dwork2016concentrated,bun2016concentrated}.

\begin{definition}[R\'enyi divergence \cite{renyi1961measures}]\label{def.renyi_divergence}
	For two probability distributions $P$ and $Q$, the \emph{R\'enyi divergence} of order $\alpha>1$ is 
	\begin{equation*}
		D_\alpha(P||Q)=\frac{1}{\alpha-1}\log \mathrm{E}_{x\sim Q} \left[\frac{P(x)}{Q(x)}\right]^\alpha.
	\end{equation*}
\end{definition}

\begin{definition}[R\'enyi DP \cite{mironov2017renyi}]\label{def.renyi_dp}
A randomized mechanism $\mathcal M$ is \emph{$(\alpha,\eps)$-R\'enyi differentially private} ($\eps$-RDP$_\alpha$)
if for all every two neighboring datasets $x$ and $x'$, $$D_{\alpha}(\mathcal M(x)||\mathcal M(x')) \leq \eps$$ 
\end{definition}

\subsection{Composition of Differentially Private Mechanisms}

A fundamental question in differential privacy concerns how privacy degrades under multiple mechanisms run on the same database. 

\subsubsection{Composition Theorems for Noninteractive Mechanisms}\label{sec.compnon}

Definition \ref{def.dp} defines differential privacy for \emph{noninteractive mechanisms} $\cM$. Composition for such noninteractive mechanisms has been extensively studied in the literature. We will denote the (noninteractive, non-adaptive) \emph{composition} of $k$ noninteractive mechanisms $\mathcal M_1, \dots, \mathcal M_k$ on a dataset $x$ as $\mathcal M := Comp(\mathcal M_1, \dots, \mathcal M_k)$, where $\calM(x)$ returns $\calM_1(x), \dots, \calM_k(x)$, with each $\calM_j$ executed independently on its own random coins.

Basic composition of $(\eps, \delta)$-differential privacy \cite{dwork2006our} finds that the privacy-loss parameters scale at most linearly with the number of mechanisms composed. Advanced composition \cite{dwork2010boosting} provides a tighter bound where the parameters scale sublinearly with the number of mechanisms, and optimal composition \cite{kairouz2015composition, murtagh2016complexity} provides an exact guarantee of this composition. DP variants such as R\'enyi DP, $f$-DP, and $z$-CDP can provide tighter composition bounds by capturing more information about the mechanisms $\cM_i$ being composed than just the two parameters $\eps_i$ and $\delta_i$. 

\subsubsection{Composition Theorems for Interactive Mechanisms}\label{sec.compinter}

While many differential privacy mechanisms are noninteractive, some mechanisms are expressly desired to be interactive, receiving and responding to adaptive queries from analysts. One example is the {\em adaptive} formulations of the aforementioned composition theorems \cite{vadhan2021concurrent, lyu2022composition, vadhan2022concurrent}. These allow an adversary to select noninteractive mechanisms adaptively based on the result received from $\calM_{i - 1}$. However, there are other, more sophisticated interactive mechanisms that answer queries in a correlated fashion based on secret state or randomness. (In adaptive composition \cite{dwork2010boosting}, each mechanism $\calM_i$ is run on the dataset using fresh randomness in dependently of how previous queries are answered.) Other examples include the Sparse Vector Technique \cite{DNRRV09,DNPR10, dwork2014algorithmic}, and Private Multiplicative Weights \cite{hardt2010multiplicative}. Thus, interactive mechanisms have been used as the basic abstraction in the programming frameworks of the open-source software project OpenDP \cite{gaboardi2020programming} as well as the Tumult Analytics platform  \cite{tumultanalyticssoftware} \cite{tumultanalyticswhitepaper}.
An interactive mechanism $\mathcal{M}$ is a party interacting with an analyst adversary in an interactive protocol, wherein each party has its \emph{random coin}, which captures the randomness used by the mechanism. 

\begin{definition}[Interactive Algorithms]
    An \emph{interactive algorithm}, also known as a {\em randomized state machine}, consists of a randomized algorithm
    $\cM : \zo^* \times \zo^* \rightarrow \zo^*\times \zo^*$ that takes the current state $s\in \zo^*$, a query $q\in \zo^*$, and returns a new state $s'\in \zo^*$ and an answer $a\in \zo^*$, written $(s',a) = \cM(s,q)$. We require that $s' \neq \lambda$ and $a \neq \lambda$, where $\lambda$ is the empty string. When we wish to make the randomness $r$ of $\cM$ more explicit, we write $(s', a) = \cM(s, q; r)$.  
\end{definition}

An interactive algorithm interacts with an analyst or adversary $\cA$ as follows.

\begin{definition}[Interaction between two mechanisms]\label{def.protocol}
	For two interactive algorithms $\cM$ and $\cA$, the \emph{interaction} between $\cM$ and $\cA$ on an input $x\in \zo^*$ to $\cM$ is the following random process (denoted $(\cA \leftrightarrow \cM(x))$):
	\begin{enumerate}
            \item Initialize $s^\cM_1 \gets x$, $m_0 \gets \lambda$, $s_1^\cA \gets \lambda$, where $\lambda$ is the empty string. 
		\item Repeat the following for $i=1,2,\ldots$.
            \begin{enumerate}
                \item If $i$ is odd, let $(s^\cM_{i+1},m_i) \gets \cM(s^\cM_i,m_{i-1})$ and $s^\cA_{i + 1} \gets s^\cA_i$.
                \item If $i$ is even, let $(s^\cA_{i+1},m_i) \gets \cA(s^\cA_i,m_{i-1})$ and $s^\cM_{i + 1} \gets s^\cM_i$.  
                \item if $m_i = \texttt{halt}$, then exit loop.
            \end{enumerate}
        \end{enumerate}
\end{definition}

In the context of DP, think of the input $x$ as a sensitive dataset that needs to be protected. In this context, an \emph{interactive mechanism} is a randomized state machine whose input space upon initialization is the space of datasets $\mathcal X$, and takes in a database $x \in \mathcal X$. An \emph{adversary} is a randomized state machine that takes in an empty string. The view of an adversary captures everything the adversary receives during the execution. 

\begin{definition}[View of the adversary in an interactive mechanism]\label{def.view}
	Let $\cM$ be an interactive mechanism and $\cA$ be an adversary interacting with the mechanism.  $\cA$'s {\em view} of $(\cA, \cM(x))$ is the tuple $\View_\cA( \cA\leftrightarrow \cM(x)) = (r_0,m_1,r_2,m_3,r_4,\ldots)$ consisting of all the messages $m_i$ received by $\cA$ from $\cM$ together with random coins $r_i$ that $\cA$ tosses when computing $m_i$ (i.e. for even $i$, $(s^\cA_{i + 1}, m_i) = \cA(s_i^\cA, m_{i - 1}; r_i)$).
\end{definition}

For shorthand, we will drop the subscript $\calA$ when referring to $\View$ for the rest of this paper, because future references to $\View$ in the rest of this paper always concern the adversary's view.

We can define interactive differential privacy based on measuring the same $(\eps,\delta)$-closeness, as in noninteractive differential privacy, between the views of the adversary on two neighboring datasets.

\begin{definition}[$(\eps, \delta)$-DP interactive mechanisms]\label{def.im_e_d}
    An interactive mechanism $\calM$ is an \emph{$(\eps, \delta)$-differentially private interactive mechanism}, or \emph{I.M.} for short, if for every pair of neighboring datasets $x, x' \in \mathcal X$, every interactive adversary algorithm $\cA$, and every subset of possible views $\mathcal S \subseteq \text{Range}(\texttt{View})$ we have $$Pr[\View(\cA \leftrightarrow \calM(x)] \in \mathcal S) \leq e^\eps Pr[\View(\cA \leftrightarrow \calM(x')) \in \mathcal S] + \delta.$$
    
\end{definition}

To capture variants such as R\'enyi DP and $f$-DP, we will define a broader version of I.M.s based on a generalized notion of differential privacy in Section \ref{sec.formaldef}.

When we consider the composition of interactive mechanisms, it is straightforward to extend the composition theorems for noninteractive mechanisms to the {\em sequential composition} of interactive mechanisms. A subtler case is the \emph{concurrent composition}, in which an adversary can interleave queries to multiple mechanisms concurrently, first studied by Vadhan and Wang \cite{vadhan2021concurrent}. The queries can therefore depend on the answers received from other mechanisms. This is useful in many practical settings. For example, an analyst may run data analyses using multiple interactive analyses on the same dataset simultaneously and the queries in multiple analyses might be correlated. Concurrent composition, formally defined in Section \ref{sec.prelim}, maintains $k$ interactive mechanisms $\calM_1, \ldots, \calM_k$. It is itself an interactive mechanism and the query it received from an adversary is of the form of $(j, q)$, meaning it issues a standard query $q$ to $\calM_j$. Concurrent composition theorems allow us to understand the privacy guarantee of the overall analysis given the privacy guarantees for each interactive analysis when they are executed independently.

Previous work provides concurrent composition theorems for several different types of differential privacy. Vadhan and Wang \cite{vadhan2021concurrent} shows that every composition theorem for noninteractive $\eps$-DP mechanisms extends to concurrent composition for interactive $\eps$-DP mechanisms. Lyu \cite{lyu2022composition} and Vadhan and Zhang \cite{vadhan2022concurrent} generalize this result to $(\eps, \delta)$-DP when $\delta > 0$.

\begin{theorem}[\cite{lyu2022composition, vadhan2022concurrent}]
Suppose that for all noninteractive mechanisms $\mathcal{M}_1, \ldots, \mathcal{M}_k$ such that $\mathcal{M}_i$ is $(\eps_i, \delta_i)$-DP for $i=1, \ldots, k$, their composition $\comp(\mathcal{M}_1, \ldots, \mathcal{M}_k)$ is $(\eps, \delta)$-DP. Then for all interactive mechanisms $\mathcal{M}_1, \ldots, \mathcal{M}_k$ with finite communication\footnote{Their proof relies on an induction argument on the number of messages exchanged, which requires an assumption of \emph{finite communication}.} such that $\mathcal{M}_i$ is $(\eps_i, \delta_i)$-DP for $i=1, \ldots, k$, 
	their concurrent composition $\concomp(\mathcal{M}_1, \ldots, \mathcal{M}_k)$, as defined formally in Algorithm \ref{def.concomp}, is $(\eps, \delta)$-DP.
 \end{theorem}

 This composition theorem also extends to $f$-DP \cite{vadhan2022concurrent}. Moreover, Lyu \cite{lyu2022composition} shows that the privacy adds up under concurrent composition for any fixed order of $\alpha>1$ for R\'enyi DP (RDP).

\begin{theorem}[\cite{lyu2022composition}]\label{thm.rdp2}
	For all $\alpha>1$, $k\in \mathbb{N}$, $\eps_1,\ldots, \eps_k>0$, and all interactive mechanisms $\mathcal{M}_1, \ldots, \mathcal{M}_k$ such that $\mathcal{M}_i$ is $(\alpha,\eps_i)$-RDP for $i=1,2 \ldots, k$, 
	their concurrent composition $\concomp(\mathcal{M}_1, \ldots, \mathcal{M}_k)$, as defined formally in Algorithm \ref{def.concomp}, is $(\alpha, \sum_{i = 1}^k \epsilon_i)$-RDP. 
\end{theorem}

In all of the above cases, privacy-loss parameters are set upfront prior to data analysis and are fixed for all queries or computations performed on the data. However, in practice, data analysts may not know in advance what they want to do with the data and may want to adaptively choose the privacy-loss parameters of subsequent mechanisms as they go along. This can lead to more efficient use of privacy resources. 
For example, we may wish to spend more privacy-loss budget on fine-grained analysis, and less privacy-loss budget on exploratory analysis. Therefore, allowing \emph{full adaptivity} where not only the mechanisms, but also the privacy-loss parameters themselves and the length of the composition can be chosen adaptively as a function of intermediate analyses, is important in practice. This consideration motivates the study of privacy \emph{filters} and \emph{odometers}, which we discuss in the following section.

\subsection{Odometers and Filters}\label{section.odoms_and_filters}

Rogers, Roth, Ullman, and Vadhan \cite{rogers2021privacy} define two primitives for the adaptive composition of DP mechanisms, privacy \emph{filters} and privacy \emph{odometers}, to allow for the ability to track privacy loss during an interaction. Their definitions were given specifically for $(\eps, \delta)$-DP; here we follow L\'ecuyer \cite{lecuyer2021practical} and work with a more general formalism that applies to arbitrary privacy measures, including $f$-DP and R\'enyi-DP. In particular, our definition of privacy \emph{odometers} is different from the odometers previously defined in \cite{rogers2021privacy}.

A privacy \emph{filter} is a mechanism that halts computation on a dataset once a preset privacy-loss budget is exceeded. It takes a global privacy-loss budget as an input and is equipped with a continuation rule that halts computation whenever the budget is exceeded. At each round, the continuation rule takes all privacy-loss parameters up to the current round, and outputs either $\texttt{continue}$ or $\texttt{halt}$. If $\texttt{continue}$, the mechanism answers the query with the current privacy parameter. Once the mechanism outputs \texttt{halt}, no further computation is allowed. (An equivalent alternative is to refuse to answer the correct query, but allow the analyst to try again with new queries. However, the halt formulation is more convenient for presentation in this paper. For notational convenience in this paper, our algorithm pseudocode will have the filter go into a looping state that always returns the same state and a dummy message once \texttt{halt} is reached.) The filter guarantees that the interaction is differentially private according to the desired privacy-loss budget.

Existing work on privacy filters and odometers has only studied cases in which the adversary adaptively chooses the next mechanism to compose with, but the mechanisms themselves are non-interactive (so there is no notion of concurrent composition, since each mechanism returns only one answer).

\begin{definition}[$\calF$-filtered composition of noninteractive $(\eps, \delta)$-DP mechanisms ($\calF \textit{-Filt(NIM)}$)]
Let $\calF$ be a continuation rule that takes in a sequence of privacy-loss parameters $(\eps_1,\delta_1), \ldots,$ and a target privacy-loss budget $(\eps, \delta)$, and maps to a binary decision: $\calF: (\mathbb{R}_{\ge 0})^2 \times (\mathbb{R}_{\ge 0})^2  \to  \{ 0, 1\}$, where $1$ means $\texttt{continue}$ and $0$ means $\texttt{halt}$, and $\mathbb R^*_{\ge 0} = \cup_{k = 0}^\infty \mathbb R^k$. The \emph{$\calF$-filtered composition of noninteractive mechanisms}, denoted as $\calF \textit{-Filt(NIM)}$, is an interactive mechanism. $\calF(\cdot; (\eps, \delta))$-$Filt(NIM)(s, m)$ is defined as follows:
\begin{enumerate}
    \item If $m = \lambda$, let $s \gets x$, and set $s' \gets (s, [])$, where $[]$ is an empty list. Return $(s', \lambda)$, where $\lambda$ is an empty string.
    \item Parse $s$ as $s= (x, [(\eps_1, \delta_1), \dots, (\eps_{k}, \delta_{k})])$
    \item If $\calM_{k + 1}$ is $(\eps_{k + 1}, \delta_{k + 1})$-DP and $\calF((\eps_1, \delta_1), \dots, (\eps_{k + 1}, \delta_{k + 1}); (\eps, \delta)) = 1$:  
    \begin{enumerate}
        \item Let $s' \gets (x, [(\calM_1, (\eps_1, \delta_1)), \dots, (\calM_{k + 1}, (\eps_{k + 1}, \delta_{k + 1}))])$
        \item Let $m' \gets \calM_{k + 1}(x)$
    \end{enumerate}
    \item Else, let $s' \gets \texttt{Halt}$, and let $m'$ store the reason for halting
    \item Return $(s', m')$
\end{enumerate}
\end{definition}

Note that at each round $k$ of computation, whether the mechanism $\calM_k$ is $(\eps_k, \delta_k)$-DP is something the implementation needs to be able to verify, e.g. by having verified privacy-loss parameters attached to every mechanism. 

\begin{definition}[Valid $(\eps,\delta)$-DP filter for noninteractive mechanisms]
We say a continuation rule $\calF$ is a 
 valid $(\eps,\delta)$-DP NIM-filter for noninteractive mechanisms if for every pair of
$(\eps, \delta)$,
$\calF(\cdot; (\eps, \delta)) \textit{-Filt(NIM)}(\cdot)$ is an $(\eps,\delta)$-DP
interactive mechanism.
\end{definition}

We define privacy \emph{odometer} as a mechanism that allows the analyst to keep track of the privacy loss at each step of computation.

\begin{definition}[Odometers]\label{def.odom_general}
    An odometer $\calO: \zo^* \times \zo^* \to \zo^* \times \zo^*$ is an I.M. that supports a specially-designated \emph{deterministic privacy-loss query}, denoted \texttt{privacy\_loss}, that is deterministic and does not change state. Upon receiving a privacy-loss query, $\calO$ will output the upper bound of the privacy loss at that point in computation.  
\end{definition}

We note that our terminology is slightly different from previous work by Rogers et al. \cite{rogers2021privacy}, where they refer to the privacy-loss accumulator as the odometer, whereas our odometer is a mechanism that has a specially-designated privacy loss query. When posing this query, the mechanism outputs the current privacy loss up to that point. 

We will stipulate that \texttt{privacy\_loss} queries are encoded as binary strings in a fixed and recognizable way, such as using all strings that begin with a $1$ to be \texttt{privacy\_loss} queries, and all strings that begin with a $0$ to be ordinary queries that can then be further parsed. 

To measure the privacy loss at each round, we use \emph{truncated view} of an adversary $\mathcal A$ interacting with an odometer $\calO$ as defined in Definition \ref{def.truncview}.
  
\begin{definition}[View between adversary $\cA$ with approximate-DP odometer $\calO$ truncated at $(\eps, \delta)$]\label{def.truncview} Given a privacy-loss parameter $(\eps, \delta)$, an adversary $\cA$, and approximate-DP odometer $\calO$, $\texttt{Trunc}_{(\eps,\delta)}(\texttt{View}(\calA \leftrightarrow \mathcal O(x))$ is constructed as follows: 

\begin{enumerate}
    \item Initialize $s^\cM_1 \gets x$, $q_0 \gets \lambda$, $s_1^\cA \gets \lambda$, where $\lambda$ is the empty string. 
    \item Repeat the following for $i=1,2,\ldots$.
        \begin{enumerate}
            \item If $i$ is odd, let  $(s^\calM_{i+1}, a_{i}) \gets \mathcal O(s^\calM_{i}, q_{i - 1})$, and let $(s_i^{\calM},(\eps_i, \delta_i)) \gets \mathcal O(s_i^{\calM}, \texttt{privacy\_loss}) $, $s^\cA_{i + 1} \gets s^\cA_i$, with the assumption that \texttt{privacy\_loss} queries do not change state. 
            \item If $\eps_i \geq \eps$ or $\delta_i \geq \delta$, then exit loop.
            \item If $i$ is even, let $(s^\calA_{i+1}, q_{i}) \gets \calA(s_{i }^\calA, a_{i - 1}; r^\calA_{i})$ and $s^\cM_{i + 1} \gets s^\cM_i$.  
            \item If $a_i = \texttt{halt}$, then exit loop.
        \end{enumerate}
    \item Return $\texttt{Trunc}_{(\eps,\delta)}(\texttt{View}(\calA \leftrightarrow \mathcal O(x)) = (r^\calA_0, a_1, \dots, r^\calA_{i-2}, a_{i - 1})$, where $i$ is the index that the for loop was halted on.
\end{enumerate}
\end{definition}

Using truncated view for analyzing odometers is first introduced by L\'ecuyer \cite{lecuyer2021practical} for RDP. We extend their definition to a general notion of truncated view, which allows us to quantify privacy loss with all other DP variants, where we can simply replace $(\eps,\delta)$ with other privacy-loss parameters for other DP variants. In contrast, the previous definition of odometer in Rogers et al. \cite{rogers2021privacy} only holds for $(\eps,\delta)$-DP.

Then the privacy guarantee of an odometer is defined as measuring the closeness between the distributions of the truncated views of an adversary. 

\begin{definition}[Valid approximate-DP odometer for noninteractive mechanisms]\label{def.eps-delta-odo}
   We say $\calO$ is a valid approximate-DP \emph{odometer} if for every pair of $(\eps,\delta)$, every adversary $\cA$, and every pair of adjacent datasets $x, x'$, 
   \begin{align*}
       D(& \texttt{Trunc}_{(\eps,\delta)}(\View(\cA\leftrightarrow \calO (x))) || \\
   & \texttt{Trunc}_{(\eps,\delta)}(\View(\cA \leftrightarrow \calO (x')))) \preceq (\eps,\delta),
   \end{align*}
   where $\preceq$ is the partial order defined by $(\eps_1, \delta_1)\preceq (\eps_2, \delta_2)$ iff $\eps_1\le \eps_2$ and $\delta_1 \le \delta_2$.
\end{definition}

A common example of an odometer is an I.M. equipped with a privacy-loss accumulator $\calG$, which is a function that gives an upper bound on the accumulated privacy loss at each step. 

\begin{definition}[Valid approximate-DP privacy-loss accumulator for noninteractive mechanisms]\label{def.eps-delta-g}
   We say $\calG$ is a \emph{valid approximate-DP NIM-privacy-loss accumulator} if $\NIO$ is a valid approximate-DP odometer.
\end{definition}

We term these odometers built from privacy-loss accumulators $\calG$-odometers. (However, a more general odometer need not be based on a privacy-loss accumulator.) Below, we provide a definition sketch for a $\calG$-odometer, which is further defined in Section \ref{sec.odometer}.

\begin{definition}[Definition sketch for $\calG$-odometer for composition of noninteractive mechanisms ($\calG \textit{-Odom(NIM)}$)]
Let $\calG$ be a privacy-loss accumulator that takes privacy-loss parameters $(\eps_1,\delta_1), \ldots, (\eps_k, \delta_k)$ in a sequence, and maps to a global privacy loss $(\eps, \delta)$, so $\calG: (\mathbb{R}^2_{\ge 0})^k \to (\mathbb{R}^2_{\ge 0})$, for every $k=1, 2 ,\ldots$. A $\calG$-odometer for composition of noninteractive mechanisms is denoted as $\calG \textit{-Odom(NIM)}$. $\calG$-$Odom(NIM)(s, m)$ is executed as follows: 
\begin{enumerate}
    \item If $m = \lambda$, let $x \gets s$, set $s' \gets (x, [])$, where $[]$ is an empty list, and return $(s', \lambda)$ 
    \item Parse $s$ as $s = (x, [(\eps_1, \delta_1), \dots, (\eps_k, \delta_k)])$
    \item If $m$ is of the form $(\calM_{k + 1}, (\eps_{k + 1}, \delta_{k + 1}))$ and $\calM_{k + 1}$ is $(\eps_{k + 1}, \delta_{k + 1})$-DP:
    \begin{enumerate}
        \item Let $s' \gets (x, [(\calM_1, (\eps_1, \delta_1)), \dots, (\calM_{k + 1}, (\eps_{k + 1}, \delta_{k + 1}))])$
        \item Let $m' \gets \calM_{k + 1}(x)$
    \end{enumerate}
    \item If $m = \texttt{privacy\_loss}$, then let $s' \gets s$ and $m' \gets \calG((\eps_1, \delta_1),\ldots, (\eps_k, \delta_k))$
    \item Return $(s', m')$.
\end{enumerate}
\end{definition}

A number of previous works construct privacy filters and odometers for noninteractive mechanisms. The original odometer definition proposed by Rogers et al.~\cite{rogers2021privacy} was defined specifically for $(\eps,\delta)$-DP and requires a simultaneous guarantee about the privacy loss being bounded at all points \emph{in time} when composing with adaptive privacy-loss parameters. Specifically, the privacy loss is defined as $\texttt{Loss}(v)=\log \left( \frac{\Pr[\View(\calA \leftrightarrow \calO(x)) =v]}{\Pr[\View(\calA \leftrightarrow \calO(x')) = v]}  \right)$, and the odometer is defined as follows.

\begin{definition}[Previous definition of Odometers in \cite{rogers2021privacy}] \label{def:RRUV-odometer}
      We say $\calG$ is a valid approximate DP \emph{odometer} if for every adversary $\cA$ in $k$-fold composition (as in, composition of $k$ mechanisms), with adaptive privacy-loss parameters, and every pair of adjacent datasets $x, x'$, the following holds with probability at most $\delta_g$ over $v\gets \View(\calA \leftrightarrow \calO(x))$
      $$\abs{\texttt{Loss}(v)} > \calG(\eps_1,\delta_1, \ldots, \eps_k,\delta_k).  $$
\end{definition}
Note that this definition yields a form of approximate DP in its final guarantee, because there is always a $\delta_g$ probability of failure.

Whitehouse, Ramdas, Rogers, and Wu~\cite{whitehouse2023fully} gave constructions of odometers and filters that quantitatively improve on the results presented in Rogers et al \cite{rogers2021privacy}, including when the composed mechanisms satisfy zCDP or R\'enyi DP (but their odometer only guarantees the same, $(\eps,\delta)$ property of Definition~\ref{def:RRUV-odometer}). 
For providing R\'enyi DP guarantees when composing R\'enyi DP mechanisms, Feldman and Zrnic \cite{feldman2022individual} constructed privacy filters.
Definitions and constructions of privacy odometers for R\'enyi DP were given by L\'ecuyer~\cite{lecuyer2021practical}. 
Our general definition of odometers follows L\'ecuyer's, which we find more natural and general, since it is applicable to all forms of DP.  This definition can be thought of as being in analogy to the definition of a $p$-value in statistics.  A $p$-value has an interpretation if we set a significance level $\alpha$ before we observe our random events, and then reject the null hypothesis if the $p$-value is larger than $\alpha$.  Similarly the privacy-loss reported by our odometer definition has an interpretation if we set a privacy-loss threshold $(\eps,\delta)$ before the mechanism is executed and halt when the odometer's privacy-loss query is not going to be smaller than $(\eps,\delta)$.  

\subsection{Our Results on Concurrent Filters and Odometers}\label{section.contrib}

In this paper, we consider privacy filters and odometers for composing \emph{interactive} mechanisms. Previous work \cite{vadhan2021concurrent, lyu2022composition, vadhan2022concurrent} have studied the concurrent composition of interactive mechanisms with pre-specified privacy-loss parameters, and in this paper we extend it to the filter and odometer settings where the privacy-loss parameters, and not just the mechanisms, can be chosen adaptively.
We analogously define $\calF$-filtered sequential composition of interactive mechanisms, or $\calF \textit{-FiltSeq(IM)}$ for short, and $\calF$-filtered concurrent composition of interactive mechanisms, or $\calF \textit{-FiltCon(IM)}$ for short. The continuation rule $\calF$ here determines whether the adversary can create new interactive mechanisms with adaptively-chosen privacy-loss parameters. Similarly, we define $\calG$-odometer for the sequential composition of interactive mechanisms, or $\calG \textit{-OdomSeq(IM)}$ for short, and $\calG$-odometer for the concurrent composition of interactive mechanisms, or $\calG \textit{-OdomCon(IM)}$ for short. It should be noted that in both our filters and odometers, the ``privacy-loss budget" is paid at the launch of each additional interactive mechanism. In particular, any \texttt{privacy\_loss} queries asked to the odometer between mechanism launches will return the same value.

As composition theorems for noninteractive mechanisms typically can be easily extended to the sequential composition of interactive mechanisms, we focus on the case of concurrent composition, as defined formally in Definition \ref{def.concomp}. 
In this setting, the adversary can interleave queries to existing interactive mechanisms. In the filter and odometer case, the adversary can also create new interactive mechanisms with adaptively-chosen privacy-loss parameters (for filters, with the added stipulation that the privacy loss with the new interactive mechanism added does not exceed the privacy budget). The complex dependence between the adversary's interacting with the different mechanisms makes it non-trivial to extend filters and odometers to this case. Nevertheless, we prove that every valid privacy filter and odometer for noninteractive mechanisms extends to a privacy filter and odometer for interactive mechanisms for $(\eps,\delta)$-DP, $f$-DP, and RDP.

\begin{definition}[Valid $(\eps,\delta)$-DP filter for interactive mechanisms]
We say a continuation rule $\calF$ is a 
 \emph{valid $(\eps,\delta)$-DP IM-filter} for interactive mechanisms if 
$\calF(\cdot; (\eps, \delta)) \textit{-FiltCon(IM)}$ is an $(\eps,\delta)$-DP
interactive mechanism.
\end{definition}

\begin{definition}[Valid ($\eps, \delta$)-DP privacy-loss accumulator for interactive mechanisms]
   We say $\calG$ is a valid \emph{approx-DP IM-privacy-loss accumulator} if $\calG$-$OdomCon(IM)$ is a valid approx-DP odometer.
\end{definition}

\begin{theorem}[$(\eps,\delta)$-DP filters and privacy-loss accumulators]\label{thm.main}
\hfill
\begin{enumerate}
    \item $\calF$ is a valid $(\eps,\delta)$-DP NIM-filter if and only if $\calF$ is a valid $(\eps,\delta)$-DP IM-filter. 
    \item $\calG$ is a valid ($\eps, \delta$)-DP NIM-privacy-loss accumulator if and only if $\calG$ is a valid ($\eps, \delta$)-DP IM-privacy-loss accumulator.
\end{enumerate}
\end{theorem}

\begin{theorem}[$f$-DP filters and privacy-loss accumulators]\label{thm.main.fDP}
\hfill
\begin{enumerate}
    \item $\calF$ is a valid $f$-DP NIM-filter if and only if $\calF$ is a valid $f$-DP IM-filter. 
    \item $\calG$ is a valid $f$-DP NIM-privacy-loss accumulator if and only if $\calG$ is a valid $f$-DP IM-privacy-loss accumulator. 
\end{enumerate}
\end{theorem}

\begin{theorem}[RDP filters and privacy-loss accumulators]\label{thm.main.RDP} 
\hfill
\begin{enumerate}
\item   $\calF(\eps_1, \eps_2, \ldots;\eps)= \mathbb{I}(\sum_i \eps_i \le \eps)$ is a valid $(\alpha, \eps)$-RDP IM-filter for every fixed order of $\alpha>1$.
\item $\calG(\eps_1, \dots, \eps_i) = \sum_{i = 1}^i \eps_i$ is a valid $(\alpha, \eps)$-RDP IM-privacy loss accumulator for every fixed order of $\alpha>1$. 
\end{enumerate}
\end{theorem}

These theorems offer strong theoretical foundations for enabling {\em full adaptivity} in composing differentially private interactive mechanisms. Firstly, they allow data analysts to adaptively choose privacy loss parameters as well as the subsequent interactive analyses. By adjusting the strength of the privacy guarantee to reflect the actual needs of the analysis as they go along, this important feature allows us to optimize the trade-off between privacy and utility. Second, our results are particularly useful for DP libraries. For example,  OpenDP\footnote{\href{https://opendp.org}{https://opendp.org}} and Tumult Labs\footnote{\href{https://www.tmlt.io}{https://www.tmlt.io}} use interactive mechanisms as a core abstraction for building differentially-private algorithms. Prior to our work, supporting adaptive selection of privacy-loss parameters has necessitated enforcing sequentiality on the use of interactive mechanisms. This restriction makes the libraries less user-friendly, as we cannot allow analysts or even DP programs interact with multiple interactive mechanisms simultaneously. Therefore, there is a gap between the existing and desired functionality. It also increases complexity in the libraries, necessitating a control system to prevent interleaving queries. 
Given the results in our paper, we can now remove this interleaving restriction, while maintaining the {\em full adaptivity} feature. We provide further discussion and our implementation for privacy filters and odometers for interactive mechanisms in Section \ref{sec.implement}.

\subsection{Proof Strategy}

In this section, we explain our proof strategy, and we defer the detailed proof to Section \ref{sec.fDP} and Section \ref{sec.RDP}. 

To prove the above theorems, our strategy is to leverage interactive postprocessing from interactive mechanisms to interactive mechanisms. We note that previous work \cite{vadhan2022concurrent} only considers interactive postprocessing from non-interactive mechanisms to interactive mechanisms. Here, we provide a general formulation called ``person-in-the-middle'', which is formally defined in Definition \ref{def.pim}.

A \emph{person-in-the-middle (PIM) mechanism} is a randomized mechanism that acts as an interlocutor between two interacting mechanisms $\cA$ and $\cM$. When $\cA$ attempts to send a message to $\cM$, the PIM mechanism can undergo an interaction with $\cA$ until finally passing on the message at the end of this interaction to $\cM$. Similarly, when $\cM$ attempts to send a message to $\cA$, the PIM mechanism can interact with $\cM$ to modify the message before passing it to $\cA$. 

This construction of the postprocessing by PIM mechanism allows us the following theorem. 

\begin{theorem}[Privacy preserved under postprocessing PIM mechanism, informally stated]\label{thm.postprocessing}
If $\cP$ is an interactive postprocessing mechanism and $\cM$ is a differentially-private interactive mechanism with respect to any privacy measure (such as $(\eps, \delta)$-DP, $f$-DP, and R\'enyi-DP), then $\cP \circ \cM$ is also differentially-private with the same privacy-loss parameters. 
\end{theorem}

A useful corollary follows from Theorem \ref{thm.postprocessing}.

\begin{corollary}
    Suppose $\mathcal N$ is an interactive mechanism over database space $\mathcal X$ such that for every pair of neighboring datasets $x, x'$ and every deterministic adversary $\cA$, there exists an $(\eps, \delta)$-DP I.M. $\mathcal M$ and an interactive postprocessing $\mathcal P$ such that 
    \begin{align*}
		\View(\calA \leftrightarrow  \mathcal{N}(x)) & \equiv \View(\calA \leftrightarrow  \mathcal{P}(\mathcal{M}(x))) \\
		\View(\calA \leftrightarrow  \mathcal{N}(x'))& \equiv \View(\calA \leftrightarrow  \mathcal{P}(\mathcal{M}(x')))
	\end{align*}

 Then $\mathcal N$ is $(\eps, \delta)$-DP. 
\end{corollary}

We will use this corollary to prove Theorem \ref{thm.main.fDP}. Our proof also relies on the following reduction theorem \cite{vadhan2022concurrent} showing that every interactive $f$-DP mechanism can be simulated by an interactive postprocessing of a noninteractive mechanism, for a fixed pair of neighboring datasets $x, x'$.

\begin{theorem}[\cite{vadhan2022concurrent}]\label{thm.reduction}
	For every trade-off function $f$, every interactive $f$-DP mechanism $\mathcal{M}$ with finite communication, and every pair of neighboring datasets $x, x'$, there exists a noninteractive $f$-DP mechanism $\mathcal{N}$ and a randomized interactive postprocessing mechanism $\mathcal{P}$ such that for every adversary $\calA$, we have 
	\begin{align*}
		\View(\calA \leftrightarrow  \mathcal{M}(x)) & \equiv \View(\calA \leftrightarrow  \mathcal{P}(\mathcal{N}(x))) \\
		\View(\calA \leftrightarrow  \mathcal{M}(x'))& \equiv \View(\calA \leftrightarrow  \mathcal{P}(\mathcal{N}(x')))
	\end{align*}
\end{theorem}

Theorem \ref{thm.reduction} allows us to reduce the composition of interactive mechanisms $\calM$ to the composition of noninteractive mechanisms $\mathcal N$, as needed to prove Theorem \ref{thm.main.fDP}. Fixing adjacent datasets $x, x'$, we can then define an interactive postprocessing $\cP$ such that for every deterministic adversary $\cA$, $\View(\cA \leftrightarrow \cP \circ \NIF) \equiv \View(\cA \leftrightarrow \CF)$ on any database $x$. Therefore we can use composition theorems for noninteractive $f$-DP mechanisms to construct concurrent $f$-DP filters. Theorem \ref{thm.main} follows directly from Theorem \ref{thm.main.fDP}, as $f$-DP captures $(\eps,\delta)$-DP as a special case. In fact, the special case of Theorem \ref{thm.reduction} for $(\eps, \delta)$-DP was shown independently by Lyu \cite{lyu2022composition}.

Our proof of Theorem \ref{thm.main.RDP} for RDP follows a different approach. Our proof relies on the concurrent composition theorem in Theorem \ref{thm.rdp2} for two interactive RDP mechanisms, and our strategy is to apply induction on the number of mechanisms being composed. 

We will make use of a \emph{universal mechanism} $\mathcal U_{\alpha, \eps}$ with parameters $(\alpha, \eps)$ be a general mechanism that can simulate any valid $(\alpha, \epsilon)$-RDP mechanism (which is specified as the first query to the universal mechanism).

When composing two mechanisms (i.e. $K = 2$ in Theorem \ref{thm.main.RDP}), by the concurrent composition theorem in Theorem \ref{thm.rdp2}, an adversary can only start a second mechanism $\calM_2$ after $\calM_1$ if and only if its privacy-loss parameter $\eps_2$ is at most $\eps-\eps_1$, where $\eps$ is the target privacy-loss budget and $\eps_1$ is the privacy-loss parameter for $\calM_1$. We define the following filter
$$
\calF_2(\eps_1, \ldots, \eps_k;\eps)=\begin{cases}
    \mathbb{I}( \sum_{i=1}^k \eps_i \le \eps ) & \text{ if } k\le 2 \\
    0 & \text{ otherwise}.
\end{cases}
$$
$\calF_2$ is a valid $(\alpha, \eps)$-RDP IM-filter for composing two mechanisms. This is because if we fix a deterministic adversary $\cA$ as Lemma \ref{lemma.deterministic_adversaries} allows us to do, then $\cM_1, \eps_1$ are predetermined adaptively, which means that the adversary interacting with $\calF_2$-$\mathit{FiltCon(IM)}$ 
is equivalent to the adversary interacting with $\concomp(\calM_1, \mathcal{U}_{\alpha, \eps - \eps_1})$, where $\mathcal{U}_{\alpha, \eps - \eps_1}$ is the universal $(\alpha, \eps-\eps_1)$-$\textit{RDP}$ mechanism, and the adversary's first query to $\mathcal{U}_{\alpha, \eps - \eps_1}$ is an $(\alpha, \eps_2)$-$\textit{RDP}$ mechanism $\mathcal{M}_2$. This allows $\mathcal{U}_{\alpha, \eps - \eps_1}$ to serve as a ``placeholder" that can simulate an adaptively chosen $\mathcal{M}_2$. Upon verifying $\eps_2\le \eps-\eps_1$, $\mathcal{U}_2$ interacts with the adversary just as $\mathcal{M}_2(x)$ does. We then apply an induction argument on the number of mechanisms being composed. For $K > 2$, we consider $\calF_K$-$FiltCon(IM)$, which is a filter for composing up to $K$ interactive mechanisms rather than two mechanisms. Assuming the privacy of $\calF_{K - 1}\text{-}FiltCon(IM)$, we argue that we can construct a postprocessing $\cP$ interacting with $\calF_2\text{-}FiltCon(IM)$ such that an adversary's interaction with $\cP \circ \calF_2\text{-}FiltCon(IM)$ is equivalent to interacting with $\calF_K$-$FiltCon(IM)$, allowing us to proceed by induction on $K$. 

Crucially, the proof for R\'enyi DP uses the fact that R\'enyi DP of a fixed order $\alpha$ is measured by a single real number $\eps$, and the optimal composition theorem is additive, i.e. $\eps = \eps_1 + \dots + \eps_K$. This is not satisfied in $(\eps, \delta)$-DP or $f$-DP, hence we use a different strategy in Theorem \ref{thm.main.fDP}, as discussed above.

Feldman and Zrnic \cite{feldman2022individual} prove that $\calF(\eps_1, \eps_2, \ldots;\eps)= \mathbb{I}(\sum_i \eps_i \le \eps)$ is a valid RDP NIM-filter for every fixed order of $\alpha>1$. They construct a supermartingale for the privacy loss, and then apply the optional stopping theorem to bound the overall privacy loss when the algorithm halts. The main idea of our proof strategy is to leverage interactivity, which allows us to proceed by induction, describing $\calF_k$ in terms of $\calF_{k-1}$. This strategy can significantly simplify the proof in \cite{feldman2022individual} of the privacy filter for even noninteractive RDP mechanisms, by similarly doing induction that uses the sequential composition of two interactive RDP mechanisms.

We then can convert this privacy filter to a valid RDP IM-privacy loss accumulator based on the following lemma.

\begin{lemma}\label{lemma.filter_to_odo}
A function $\calG$ is a valid RDP IM-privacy loss accumulator if and only if $\calF(\cdot;\eps) = I(\calG(\cdot)\leq \eps)$ is a valid RDP IM-filter for all $\eps > 0$. 
\end{lemma}

We give a generalized version of Lemma \ref{lemma.filter_to_odo} in Section \ref{sec.formaldef}, meaning that an analogue of Lemma \ref{lemma.filter_to_odo} also holds for all other variants of DP. 



\section{Concurrent Filter and Odometer}\label{sec.formaldef}

\subsection{Preliminaries}\label{sec.prelim}

Vadhan and Zhang \cite{vadhan2022concurrent} define the generalized notion of privacy, termed $\calD$ DP, based on the following generalized probability distance. The distance is defined on a partially ordered set, which allows us to compare privacy guarantees at different protection levels. It should enjoy the postprocessing and joint convexity properties.

\begin{definition}[Generalized probability distance \cite{vadhan2022concurrent}]\label{def.genprobdistance}
	A {\em generalized probability distance measure} is a tuple $(\mathcal{D}, \preceq, D)$ such that
	\begin{enumerate}
		\item $(\mathcal{D}, \preceq)$ is a partially ordered set (poset).
		\item $D$ is a mapping that takes any two random variables $X,X'$ over the same measurable space to an element $D(X,X')$ of $\mathcal{D}$.
		\item (Postprocessing.) The generalized distance mapping $D$ is closed under postprocessing, meaning that for every measurable function $g$, $D(g(X), g(X'))\preceq D(X,X')$.
		\item (Joint Convexity.)  Suppose we have a collection of random variables $(X_i, X'_i)_{i\in \mathcal{I}}$ and a random variable $I$ distributed on $\mathcal{I}$. If $D(X_i, X'_i) \preceq d$ for all $i\in \mathcal{I}$, then $D(X_I, X'_I) \preceq d$.
	\end{enumerate}
\end{definition}

With this definition, the difficulty of distinguishing two neighboring datasets is measured by the generalized distance between the distributions of an adversary's views. 

\begin{definition}[$\calD$ DP \cite{vadhan2022concurrent}]\label{def.ddp}
	Let $(\mathcal{D}, \preceq, D)$ be a generalized probability distance.
	For $d \in \mathcal{D}$, we call an interactive mechanism $\mathcal{M}$ {\em $d$-$\calD$ DP} if for every interactive algorithm $\mathcal A$ and every pair of neighboring datasets $x,x'$, we have 
	$$D(\View ( \mathcal A \leftrightarrow \mathcal{M}(x))   ,\View (\mathcal A \leftrightarrow \mathcal{M}(x') )) \preceq d.$$ 
\end{definition} 


A convenient property of $\calD$ DP is that it suffices to consider deterministic adversaries, because $\calD$ DP satisfies joint convexity. Since the filters and odometers we define will be $\calD$ DP, we will only consider deterministic adversaries in this paper. 

\begin{lemma}[Deterministic adversaries for $d$-$\calD$ DP \cite{vadhan2021concurrent}\cite{vadhan2022concurrent}]\label{lemma.deterministic_adversaries}
    An interactive mechanism M is $d$-$\calD$ DP, if and only if for every pair of neighboring datasets $x, x'$, for every deterministic adversary $\mathcal A$, $$D(\texttt{View}(\mathcal A \leftrightarrow \calM(x))||\texttt{View}(\mathcal A \leftrightarrow \calM(x'))) \preceq d.$$
\end{lemma}

For the rest of this paper, we will notate $\calD^* = \bigcup_{k=0}^\infty \calD^k$.

Finally, we define concurrent composition formally. In this setting, the adversary can arbitrarily and adaptively interleave queries between several differentially-private mechanisms, meaning that the queries can be dependent.

\begin{definition}[Concurrent composition of interactive mechanisms \cite{vadhan2021concurrent}] \label{def.concomp}
    Let $\calM_1, \dots, \calM_k$ be interactive mechanisms taking private datasets $x_1, \dots, x_k$ respectively. The concurrent composition of $\calM_1, \dots, \calM_k$, denoted $\calM = \concomp(\calM_1, \dots, \calM_k)$, is defined as in Algorithm \ref{alg.concomp}.
\end{definition}

\SetKwProg{Fn}{def}{:}{}
\SetKwFunction{AlgM}{$\calM$}%

\begin{algorithm}
    \DontPrintSemicolon
    \caption{Concurrent composition of interactive mechanisms}\label{alg.concomp}
        \Fn{\AlgM{s, m}}{
        \If{$m = \lambda$}{
            $s \gets (x, [(\calM_1, s_1), \dots, (\calM_k, s_k)])$ \tcp*[f]{initialize $k$ mechanisms}
            \Return{$(s, \lambda)$}\;
        }
        Parse $s = (x, [(\calM_1, s_1),\dots, (\calM_k, s_k)])$\;
        Parse $m$\;
        \uIf{$m$ is of the form $m = (j, q)$ where $j = 1, \dots, k$ and $q$ is a query to $\calM_j$}{
            $(s_j', m') \gets \calM_j(s_j, q)$\;
            $s' \gets (x, [(\calM_1, s_1), \dots, (\calM_j, s_j'), \dots, (\calM_k, s_k)])$\;
        }\ElseIf{$m$ cannot be parsed correctly}{
            $s' \gets s$, $m' \gets \texttt{invalid query}$\;
        }
        \Return $(s', m')$\;
        }
\end{algorithm}

\subsection{Filters}

We can now define a generalization of the privacy filter for the composition of noninteractive mechanisms as introduced in Section \ref{section.odoms_and_filters}. Recall that filters are based on continuation rules. In line with the fully-adaptive setting, the number of interactive mechanisms need not be specified beforehand.

\begin{definition}[$\calF$-filtered composition of $\calD$ DP noninteractive mechanisms ($\NIF$)]\label{def.f_noninteractive}
     For a continuation rule $\calF: \calD^* \times \calD \to \{ 0, 1\}$, the $\calF$-filtered composition of $\calD$ DP interactive mechanisms with privacy-loss budget $d$ is the interactive mechanism given in Algorithm \ref{alg.k_noninteractive_comp}.
\end{definition}

If there is a finite maximum $K$ of noninteractive privacy-loss parameters $d_1, d_2, \dots, d_K$ to take in, we denote the continuation rule $\calF_K(\cdot; d): \calD^* \times \calD \rightarrow \{ 1, 0 \}$, where $d$ is the target privacy budget.

\SetKwFunction{AlgNIF}{$\calF(\cdot; d)\text{-}Filt(NIM)$}%
 
\begin{algorithm}
    \DontPrintSemicolon
    \caption{$\calF$-filtered composition of $\calD$ DP noninteractive mechanisms ($\NIF$), for $\calF_K: \calD^* \to \calD$, where $K$ could be $\infty$.} \label{alg.k_noninteractive_comp}

    \Fn{\AlgNIF{$s, m$}}{ 
        \If{$m = \lambda$}{
            $x \gets s$, $s' \gets (x, [])$, where $[]$ is an empty list \tcp*[f]{initialize filter}\;
            \Return $(s', \lambda)$\;
        }
        Parse $s$ as $s = (x, [(\calM_1, d_1), (\calM_2, d_2), \dots, (\calM_k, d_k)])$\;
        Parse $m$\;
        \uIf{$m$ is of the form $m = (\calM', d')$}{
            \uIf{$k = K$}{
                $s' \gets \texttt{Halt}$, $m' \gets \texttt{at mechanism count limit}$\;
            }\ElseIf{$\calM'$ is $d'$-$\calD$ DP}{
                $\calM_{k + 1} \gets \calM'$\;
                $d_{k + 1} \gets d'$\;
            }
            \If{$\calF(d_1, \dots, d_{k + 1}) \neq 1$}{
                $m' \gets \texttt{insufficient budget}$, $s' \gets \texttt{Halt}$\;
            }\Else{
                Random sample $r_{k + 1} :=$ coin tosses for $\calM_{k + 1}$\;
                $s' \gets (x, [d_1, \dots, d_{k + 1}])$\;
                $m' \gets \calM_{k + 1}(x; r_{k + 1})$\;
            }
        }\ElseIf{$m$ cannot be parsed correctly}{
            $s' \gets s$, $m' \gets \texttt{invalid query}$\;
        }
        \Return $(s', m')$\;
        }
\end{algorithm}

\begin{definition}[Valid $\calD$ DP filter for noninteractive mechanisms]\label{def.ddp_f}
   Let $\calF$ be a continuation rule $\calF: \calD^*\times \calD \rightarrow \{0, 1\}$. We say $\calF$ is a valid
$\calD$ DP \emph{NIM-filter} for noninteractive mechanisms if for every
$d\in \calD$,
$\calF(\cdot; d)$-$Filt(NIM)$ is a $d$-$\calD$ DP
interactive mechanism.
\end{definition}

We can now similarly define the privacy filter for concurrent composition of interactive mechanisms. 

\begin{definition}[$\calF$-filtered concurrent composition of $\calD$ DP interactive mechanisms]\label{def.cf_interactive}
The $\calF$-filtered concurrent composition of $\calD$ DP interactive mechanisms, denoted as $\CF$, is an interactive mechanism as executed in Algorithm \ref{alg.cf_interactive}.
\end{definition}

\SetKwFunction{AlgCF}{$\calF(\cdot; d)\text{-}FiltCon(IM)$}%

\begin{algorithm}
    \DontPrintSemicolon
    \caption{$\calF$-filtered concurrent composition of $\calD$ DP interactive mechanisms ($\CF$), for $\calF_K: \calD^* \to \calD$, where $K$ could be $\infty$.} \label{alg.cf_interactive}
    \Fn{\AlgCF{s, m}}{
        \If(\tcp*[f]{initialize filter}){$m = \lambda$}{
            $x \gets s$, $s' \gets (x, [])$, where $[]$ is an empty list\; 
            \Return $(s', \lambda)$\;
        }
        Parse $s$ as $s = (x, [(\calM_1, d_1, s_{1}), \dots, (\calM_k, d_k, s_{k})])$\;
        Parse $m$\;
        \If{$m$ is of the form $m = (\calM', d')$}{
            \If(\tcp*[f]{already at mechanism limit}){$k = K$}{
                $s' \gets \texttt{Halt}$, $m' \gets \texttt{at mechanism count limit}$\;
            }\ElseIf{$\calM'$ is $d'$-$\calD$ DP}{
                $\calM_{k + 1} \gets \calM'$\;
                $d_{k + 1} \gets d'$\;

                \If{$\calF(d_1, \dots, d_{k + 1}) \neq 1$}{
                    $s' \gets \texttt{Halt}$, $m' \gets \texttt{insufficient budget}$\;
                }\Else{
                    $(s_{k + 1}, m') \gets \cM_{k + 1}(x, \lambda)$\;
                    $s' \gets (x, [(\calM_1, d_1, s_{1}), \dots, (\calM_{k + 1}, d_{k + 1}, s_{k + 1})])$\;
                }
            }
        }\ElseIf{$m$ is of the form $m = (j, q)$ where $j = 1, \dots, k$ and $q$ is a query to $\calM_j$}{
            $(s_{j}', m') \gets \calM_j(s_{j}, q)$\;
            $s' \gets (x, d, [(\calM_{1}, d_{1}, s_{1}), \dots, (\calM_{j}, d_{j}, s_{j}'), \dots, (\calM_{k}, d_{k}, s_{k})])$\;
        }\ElseIf{$m$ cannot be parsed correctly}{
            $s' \gets s$, $m' \gets \texttt{invalid query}$\;
        }
        \Return{$(s', m')$}
    }
\end{algorithm}

A privacy filter is one in which Definition \ref{def.cf_interactive} holds for all distances $d \in \calD$.
\begin{definition}[Valid $\calD$ DP filter for interactive mechanisms]\label{def.ddp_cf}
     Let $\calF: \calD^*\times \calD \rightarrow \{0, 1\}$ be a continuation rule. $\calF$ is a valid \emph{$\calD$ DP concurrent IM-filter} if for every $d\in \calD$, $\calF(\cdot; d)$-$FiltCon(IM)$ is a $d$-$\calD$ DP interactive mechanism.
\end{definition}

\subsection{Odometer}\label{sec.odometer}

We can similarly generalize our definition for privacy odometers to $\calD$ DP. 

\begin{definition}[$\calG$-odometer for composition of noninteractive mechanisms ($\NIO$)]
    A $\calG$-odometer for the composition of noninteractive $\calD$ DP mechanisms, denoted as \\ $\NIO$, is executed as in Algorithm \ref{alg.o_construction}.
\end{definition} 

\SetKwFunction{AlgNIO}{$\NIO$}

\begin{algorithm}
    \DontPrintSemicolon
    \caption{$\calG$-odometer for composition of noninteractive mechanisms ($\NIO$) .} \label{alg.o_construction}
    \Fn{\AlgNIO{s, m}}{
        \If{$m = \lambda$}{
            $x \gets s$, $s' \gets (x, [])$, where $[]$ is an empty list\;
            \Return{$(s', \lambda)$}\;
        }
        Parse $s$ as $s = (x, [d_1, d_2, \dots, d_k])$\;
        Parse $m$\;
        \If{$m$ is of the form $m = (\calM', d')$}{
            \If{$\calM'$ is $d'$-$\calD$ DP}{
                $\calM_{k + 1} \gets \calM'$\;
                $d_{k + 1} \gets d'$\;
                $s' \gets (x, [d_1, \dots, d_{k + 1}])$\;
                Random sample $r_{k + 1} :=$ coin tosses for $\cM_{k + 1}$\;
                $m' \gets \calM_{k + 1}(x; r_{k + 1})$\;
            }\Else{
                $s' \gets s$, $m' \gets \texttt{Divergence cannot be parsed}$\;
            }
        }\ElseIf{$m = \texttt{privacy\_loss}$}{
            $s' \gets s, m' \gets \calG(d_1, \dots, d_k)$\;
        }\ElseIf{$m$ cannot be parsed correctly}{
            $s' \gets s$, $m' \gets \texttt{invalid query}$\;
        }
        \Return{($s', m'$)}\;
     }
\end{algorithm}

Odometers come with a per-mechanism privacy guarantee, meaning that the divergence between the views after each new mechanism is started should be at most some set distance $d \in \calD$ apart. Because the odometer can receive a \texttt{privacy\_loss} query at any point in computation, a notion that captures the during-computation guarantee is necessary. The generalized \emph{truncated view} of the adversary up to the $n$-th round of interaction to an odometer equipped with continuation rule $\calG$ is defined as in Algorithm \ref{alg:truncview}.

\begin{definition}[Truncated view for odometers \cite{lecuyer2021practical}]
    For a deterministic adversary $\calA$ and odometer $\calO$ on dataset $x$, the \emph{truncated view} $$\texttt{Trunc}_d (\View(\calA \leftrightarrow \calO(x)))$$ is the return value of Algorithm \ref{alg:truncview}, including the randomness and query answers.
\end{definition}

\SetKwFunction{AlgTruncView}{TruncView}

\begin{algorithm}
    \DontPrintSemicolon
    \caption{Truncated view of adversary interacting with odometer $\calO$ given input $d$}\label{alg:truncview}
    
    \Fn{\AlgTruncView{$\View(\cA \leftrightarrow \calO(x))$}}{
        Initialize $s_1^\cM \gets x$, $q_0 \gets \lambda$, $s_1^\cA \gets \lambda$\;
        \For{$i = 1, 2, \dots$} {
            \If{$i$ is odd} {
                $(s^\cM_{i + 1}, a_i) \gets \calO(s^\cM_i, q_{i - 1})$\;
                $(s^\cM_i, d_i) = \calO(s^\cM_i, \texttt{privacy\_loss})$\;
                $s^\cA_{i + 1} \gets s^\calA_i$\;
                \If{$d_i \not\preceq d$}{
                    Exit loop
                }
            }\ElseIf{$i$ is even}{
                $(s^\cA_{i + 1}, q_i) \gets \cA(s_i^\cA, a_{i - 1}; r_i^\cA)$\;
                $s^\cM_{i + 1} \gets s^\cM_i$\;
            }

            \If{$s_{i + 1}^\cM = \texttt{Halt}$ or $s_{i + 1}^\cA = \texttt{Halt}$} {
                Exit loop\;
            }
        }
        \Return{$\texttt{Trunc}_{d}(\texttt{View}(\calA \leftrightarrow \mathcal O(x))) = (r_0^\calA, a_1, r_2^\calA, a_3, \dots, r_{i - 2}^\cA, a_{i - 1})$}
    }
\end{algorithm}

Having generalized truncated view to $\calD$ DP, we can now define a \emph{valid} $\calD$ DP odometer, and a valid privacy-loss accumulator for an odometer of noninteractive mechanisms.

\begin{definition}[Valid $\mathcal D$-DP odometer]
    $\calO$ is a \emph{valid $\calD$ DP odometer} if for every $d \in \calD$, adversary $\mathcal A$, and every pair of adjacent datasets $x, x'$, 
    \begin{align*}
        D(\texttt{Trunc}_d(\View(\mathcal A\leftrightarrow \calO(x))) || \texttt{Trunc}_d(\View(\mathcal A \leftrightarrow \calO(x')))) \preceq d
    \end{align*}
\end{definition}

\begin{definition}[Valid $\mathcal D$-DP privacy-loss accumulator for noninteractive mechanisms]
    $\calG: \calD^* \to \calD$ is a valid $\calD$ DP noninteractive-privacy-loss accumulator (NIM-privacy-loss-accumulator) if for every $d \in \calD$, $\calG(\cdot)$-$Filt(NIM)$ is a valid $d$-$\calD$ DP odometer. 
\end{definition}

We now define the privacy odometer that allows us to keep track of the total privacy loss across multiple interactive mechanisms, where the adversary can interleave queries to existing interactive mechanisms, as well as create new interactive mechanisms with privacy-loss budgets chosen adaptively. Crucially, our construction of odometers have a specially-designated privacy-loss query, which outputs a conservative upper bound on the privacy loss for the interaction up to that point. It should be noted that this odometer model for I.M.s does not define a per-query privacy loss as in the odometer for noninteractive mechanisms. Instead, the ``privacy-loss budget" is paid at the launch of each additional interactive mechanism. Any \texttt{privacy\_loss} queries asked between mechanism launches will return the same value.

\begin{definition}[$\calG$-odometer for concurrent composition of interactive mechanisms ($\CO$)]\label{def.co_construction}
A $\calG$-odometer for the \emph{concurrent composition of interactive mechanisms}, denoted as $\CO$, is defined as in Algorithm \ref{alg.co_construction}.
\end{definition}

\SetKwFunction{AlgCO}{$\CO$}

\begin{algorithm}
    \DontPrintSemicolon
    \caption{$\calG$-odometer for concurrent composition of interactive mechanisms ($\CO$).} \label{alg.co_construction}
    \Fn{\AlgCO{$s, m$}}{
        \If{$m = \lambda$}{
            $x \gets s$, $s' \gets (x, [])$, where $[]$ is an empty list\;
            \Return{$(s', \lambda)$}\;
        }
        Parse $s = (x, [(\calM_1, d_1, s_{1}), (\calM_2, d_2, s_{2}), \dots, (\calM_k, d_k, s_{k})])$\;
        Parse $m$\;
        \If{$m$ is of the form $m = (\calM', d')$}{
            \If{$\calM'$ is $d'$-$\calD$ DP}{
                $\calM_{k + 1} \gets \calM'$\;
                $d_{k + 1} \gets d'$\;
                $(s_{k+1}, m') \gets \cM_{k + 1}(x, \lambda)$\;
                $s' \gets (x,  [(\calM_1, d_1, s_{1}), \dots, (\calM_{k + 1}, d_{k + 1}, s_{k + 1})])$\;
            }
            \Else{
                $s' \gets s$, $m' \gets \texttt{``Divergence cannot be parsed"}$\;
            }
        }\ElseIf{$m$ is of the form $m = (j, q)$ where $j = 1, \dots, k$ and $q$ is a query to $\calM_j$}{
            $(s_{j}', m') \gets \calM_j(s_{j}, q)$\;
            $s' \gets (x, [(\calM_1, d_1, s_1), \dots, (\calM_{j}, d_{j}, s_{j}'), \dots, (\calM_k, d_k, s_k)])$\;
        }\ElseIf{$m = \texttt{privacy\_loss}$}{
            $s' \gets s, m' \gets \calG(d_1, \dots, d_k)$\;
        }\ElseIf{$m$ cannot be parsed correctly}{
            $s' \gets s$, $m' \gets \texttt{invalid query}$\;
        }
        \Return{$(s', m')$}\;
    }
\end{algorithm}

We are now ready to define the notion of a \emph{valid} $\calD$ DP privacy-loss accumulator. 

\begin{definition}[Valid $\calD$ DP privacy-loss accumulator for interactive mechanisms]\label{def.ddp_odom}
We say $\calG$ is a valid $\calD$ DP \emph{concurrent IM-privacy-loss accumulator} if for every $d\in \calD$, $\calG(\cdot; d)$-$OdomCon(IM)$ is a valid $d$-$\calD$ DP odometer. 
\end{definition}

\subsection{I.M.-to-I.M. Interactive Postprocessing}

Before we move on to defining concurrent filters and odometers, we will formulate interactive postprocessing in terms of state machines, which will become a helpful tool for us to prove important theorems related to concurrent filters and odometers for $f$-DP and R\'enyi-DP mechanisms. We will first give a formal definition of person-in-the-middle mechanisms. 

\begin{definition}[Person-in-the-middle mechanism]\label{def.pim}  An {\em interactive postprocessing mechanism} (or `PIM mechanism' for `person-in-the-middle') is a randomized mechanism $\cP : \zo^* \times \{Q,A\} \times \zo^* \rightarrow \zo^* \times \{Q,A\} \times \zo^*$.  $\cP$ takes its current state $s \in \zo^*$, a value $v\in \{Q,A\}$ to indicate whether it is receiving a query (from an analyst or adversary) or an answer (from the mechanism it is postprocessing), and a message $m$ (which is either a query or an answer) and returns a new state $s'\in \zo^*$, a value $v'\in \{Q,A\}$ to indicate whether it is asking a query (of the mechanism it is postprocessing) or providing an answer (if it is answering a query) and message $m'$ (which is either a query or an answer), denoted $(s',v',m')=\cP(s,v,m)$.
\end{definition}

With the PIM primitive, we can now more precisely define a postprocessed interactive mechanism. 

\begin{definition}[Postprocessed interactive mechanism]
Let $\cM$ be an interactive mechanism and $\cP$ be an interactive postprocessing. Then
the postprocessing of $\cM$ by a PIM algorithm  $\cP$, 
denoted $\cP\circ \cM$, is the interactive mechanism defined in Algorithm \ref{alg:p_circ_m}.
\end{definition}

\SetKwFunction{AlgPost}{$(\cP \circ \cM)$}

\begin{algorithm}
    \DontPrintSemicolon
    \caption{Postprocessing of $\cM$ by PIM algorithm $\cP$, denoted $\cP \circ \cM$.}\label{alg:p_circ_m}
        \Fn{\AlgPost{$s, m$}}{
            \If{$m = \lambda$}{
                $s^\cM \gets s$, $s^\cP \gets \lambda$, $m \gets \lambda$, $v \gets Q$ \tcp*[f]{initialize $\cM$ and $\cP$}
            }\Else{
                 Parse $s$ as $s = (s^\cM, s^\cP)$\;
                 Let $(s^\cP, v, m) \gets \cP(s^\cP, Q, m)$\;
            }
            \While(\tcc{while $\cP$ is equipped to interact with $\cM$, conduct interaction}){$v = Q$}{
                $(s^\cM, m) \gets \cM(s^\cM, m)$\;
                $(s^\cP, v, m) = \cP(s^\cP, A, m)$\;
            }
             \Return $((s^\cM, s^\cP), m)$\;
        }
\end{algorithm}

\begin{theorem}[Privacy holds under I.M.-to-I.M. postprocessing]\label{thm.pim_privacy}
If $\cP: \zo^* \times \{Q,A\} \times \zo^* \rightarrow \zo^* \times \{Q,A\} \times \zo^*$ is an interactive postprocessing algorithm and $\cM$ is a $d$-$\calD$ DP interactive mechanism, then $\cP\circ \cM$ is $d$-$\calD$ DP.
\end{theorem}

\begin{proof}
Let $\cA$ be a deterministic adversary against $\cP \circ \cM$. By Definition \ref{def.view}, $\View( \cA\leftrightarrow \cP \circ \cM(x))  = (r_0, m_1, r_2, m_3, r_4, \dots)$, where $r_i$ are the random coins that $\cA$ tosses and $m_i$ are the messages received by $\cA$. Define an adversary $\cA'$ that maintains the states of $\cA$ and $\cP$ as submachines, constructed as in Algorithm \ref{alg.im_to_im_adversary}. $\cA'$ is a normal adversary, so it will always get started by receiving an answer and end by making a query. 

Informally, Algorithm \ref{alg.im_to_im_adversary} says that when $\cA'$ is to generate a query (determined by $v = Q$), it will run $\cA$ with $\cP$ to generate a query. When $\cA'$ is to receive an answer (determined by $v = A$), it will run $\cP$ to first check if $\cP$ wants to interactively process it with $\cM$ first. If so, the interaction between $\cP$ and $\cM$ would proceed until $\cP$ determines that the answer is in a state to be passed to $\cA$, done by checking the output value $v_{out}$ of $\cP$ given the state, the value $A$, and the message as inputs at that round.

\SetKwFunction{AlgA}{$\calA'$}

\begin{algorithm}
    \DontPrintSemicolon
    \caption{Adversary $\calA'$ that maintains the states of $\cA$ and $\cP$ as submachines.} \label{alg.im_to_im_adversary}
        \Fn{\AlgA{$s, m$}}{
            \If(\tcp*[f]{need to initialize $\cA'$}){$s = \lambda$} {
                $(s^\cA, m_0) \gets \cA(\lambda, \lambda)$ \tcp*[f]{initialize adversary submachine}
                $(s^\cP, v, m') \gets \cP(\lambda, A, m)$ \tcp*[f]{initialize PIM submachine}
            }\Else{
                Parse $s = (s^\cA, s^\cP)$\;
            }
            $v_{out} \gets A$\tcp*[f]{initialize $v_{out}$}
            \If(\tcp*[f]{$\cP$ will provide a query to $\cM$}){$v = Q$}{
                \While{$v_{out} = A$}{
                    Let $((s^\cA)', \_) \gets \cA(s^\cA, m)$\;
                    Let $((s^\cP)', v_{out}, m') \gets \cP(s^\cP, v, \_)$\;
                    $s^\cA \gets (s^\cA)'$, $s^\cP \gets (s^\cP)'$, $m \gets m'$\;
                }
                $s' \gets ((s^\cA)', (s^\cP)')$\;
                $v \gets A$ \tcp*[f]{prime $\cA$ to receive an answer}
            }\ElseIf{$v = A$}{
                $(s^\cP, v_{out}, m') \gets \cP(s^\cP, v, m)$\;
                $s' \gets (s^\cA, s^\cP)$\;
                \If(\tcp*[f]{$\cP$ will provide an answer to $\cA$}){$v_{out} = A$}{
                    $v \gets Q$ \tcp*[f]{prime $\cA$ to ask a query}
                }
            }
            \Return{$(s', m')$}
        }
\end{algorithm}

Because $\cM$ is $d$-$\calD$ DP, the interaction between $\cA'$ and $\cM$ is $d$-$\calD$ DP. Define a noninteractive postprocessing function $g$ that transforms $\View(\cA' \leftrightarrow \cM(x)) = (r_0', m_1', r_2', m_3', \dots)$ on any dataset $x$ as in Algorithm \ref{alg.ddp_postprocessing}. Informally, the new view object generated by Algorithm \ref{alg.ddp_postprocessing} makes visible the submachine interaction between $\cA$ and $\cP$ in the operation of $\cA'$, and snips the interaction between $\calA'$ and $\cM$ such that if there is an extended sequence between $\cM$ and submachine $\cP$ of $\cA$, only the final answer from $\cM$ before $\cP$ and $\cA$ begin interacting will be included. Since noninteractive postprocessing preserves privacy properties by Definition \ref{def.genprobdistance} (3), we know that 
$$D(g(\View(\cA' \leftrightarrow \cM)(x)) || g(\View(\cA' \leftrightarrow \cM)(x'))) \preceq D(\View(\cA' \leftrightarrow \cM)(x) || \View(\cA' \leftrightarrow \cM)(x')) \preceq d.$$ 

\SetKwFunction{AlgDDPpost}{$g$}

\begin{algorithm}
    \DontPrintSemicolon
    \caption{Noninteractive postprocessing $g$ that transforms $\View(\cA' \leftrightarrow \cM(x))$ to $\View(\cA \leftrightarrow (\cP \circ \cM)(x))$ 
    }\label{alg.ddp_postprocessing}
        \Fn{\AlgDDPpost{$\View(\calA' \leftrightarrow \cM(x))$}}{
            $i \gets 1$ \tcp*[f]{$i$ indexes input}
            $j \gets 0$ \tcp*[f]{$j$ indexes output}
            Parse $\View(\calA' \leftrightarrow \cM(x))$ as $(r_0', m_1', r_2', m_3', \dots)$\;
            $v \gets Q, v_{out} \gets A$\;
            $(s^\cA, m^\cA) \gets \cA(\lambda, \lambda)$ \;
            $(s^\cP, v, m^\cP) \gets \cP(\lambda, v, m'_1)$\;
            \While{$m'_i$} {
                \If(\tcp*[f]{adding interaction $\cA \leftrightarrow \cP$}){$v = Q$}{
                \While(\tcp*[f]{simulate interaction of $\cA$ and $\cP$}){$v_{out} \neq Q$} {
                    $(s^\cA, m^\cA) \gets \cA(s^\cA, m^\cA; r^\cA_j)$\;
                    $(s^\cP, v_{out}, m^\cP) \gets \cP(s^{\cP}, v, m^\cA; r^\cP_{j})$ \tcp*[f]{using the randomness of $\cP$ to get the new state}
                    $m_{j + 1} \gets m^\cP$ \tcp*[f]{record message $\cA$ receives from $\cP$}
                    $j \gets j + 2$\;
                }
                $v \gets A$\;
                }\ElseIf(\tcp*[f]{abbreviating interaction $\cP \leftrightarrow \cM$}){$v = A$}{
                    $(s^\cP, v_{out}, m) \gets \cP(s^{\cP}_i, v, m_i'; r^\cP_{j})$\;
                    \If{$v_{out} = Q$}{
                        $v \gets A$\;
                    }\ElseIf(\tcp*[f]{record the history}){$v_{out} = A$}{
                        $m_{j + 1} \gets m$\;
                        $v \gets Q$\;
                        $j \gets j + 2$\;
                    }
                }
                $i \gets i + 2$\;
            }
            \Return{$(r_0^\cA, m_1, r_2^\cA, m_3, \dots, r_j^\cA, m_{j + 1}, \dots)$}\;
        }
\end{algorithm}

The view of $\cA$ against $\cP \circ \cM$ on any dataset $x$ can be computed by postprocessing the view of $\cA'$ against $\cM$ on $x$ using $g$; in other words, $\View(\cA \leftrightarrow \cP \circ \cM) \equiv g(\View(\cA' \leftrightarrow \cM))$ on any dataset $x$. Therefore, $D(\View(\cA \leftrightarrow \cP \circ \cM)(x) || \View(\cA \leftrightarrow \cP \circ \cM)(x')) \preceq d$. This means $\cP \circ \cM$ is also $d$-$\calD$ DP.
\end{proof}


A useful corollary follows directly from Theorem \ref{thm.pim_privacy}.

\begin{corollary}\label{corollary.pim_privacy}
    Suppose $\mathcal N: \zo^* \times \zo^* \to \zo^* \times \zo^*$ is an interactive mechanism such that for every pair of neighboring datasets $x, x'$, every deterministic adversary $\calA$, there exists an interactive mechanism $\calM$ that is $d$-$\calD$ DP on dataset universe $\{ x, x' \}$ and a PIM mechanism $\cP$ such that 
    \begin{align*}
        \View(\mathcal A \leftrightarrow \mathcal N(x)) & \equiv \View(\mathcal A \leftrightarrow \mathcal P \circ \calM(x)) \\
        \View(\mathcal A \leftrightarrow \mathcal N(x')) & \equiv \View(\mathcal A \leftrightarrow \mathcal P \circ \calM(x'))
    \end{align*}
\end{corollary}

\begin{proof}

By Theorem \ref{thm.pim_privacy}, $\cP \circ \cM$ is $d$-$\calD$ DP on dataset universe $\{ x, x' \}$. Therefore, $D(\View(\mathcal A \leftrightarrow \mathcal N(x))||\View(\mathcal A \leftrightarrow \mathcal N(x'))) = D(\View(\mathcal A \leftrightarrow \mathcal P \circ \calM(x))||\View(\mathcal A \leftrightarrow \mathcal P \circ \calM(x'))) \preceq d$. Therefore, $\mathcal N$ is also $d$-$\calD$ DP.  
\end{proof}

A particularly convenient property of privacy-loss accumulators and privacy filters is that a valid privacy-loss accumulator can be converted into a valid privacy filter, and vice versa. This bijective property allows us to use a filter to build an odometer, and an odometer to build a filter of the form $\mathbb I(\calG(\cdot) \preceq d)$. Any nice properties of one would hold for the other. 

\begin{lemma}\label{lemma.f_odom}
\begin{enumerate}
    \item A function $\calG:\calD^*\rightarrow \calD$ is a valid NIM-privacy-loss accumulator if and only if $\calF: \calD^* \times \calD \to \{0, 1\}$ constructed from $\calG(\cdot)$ such that $\calF(\cdot;d) = \mathbb I(\calG(\cdot) \preceq d)$ is a valid $\calD$-DP NIM-filter.
    \item A function $\calG:\calD^*\rightarrow \calD$ is a valid concurrent IM-privacy-loss accumulator if and only if $\calF: \calD^* \times \calD \to \{0, 1\}$ constructed from $\calG(\cdot)$ such that $\calF(\cdot;d) = \mathbb I(\calG(\cdot) \preceq d)$ is a valid $\calD$-DP concurrent IM-filter.
\end{enumerate}
\end{lemma}

The full proof of this theorem is enclosed in  Appendix \ref{appendix.f_odom_proof}. 
The approach is to first define an I.M. $\calM$ that such that for any adversary $\calA$ and dataset $x$, $\View(\calA \leftrightarrow \calM(x))$ is exactly the same as $\texttt{Trunc}_d(\View(\calA \leftrightarrow \CO(x, \cdot)))$. 
Then, using Theorem \ref{thm.pim_privacy}, we can show that $\cM$ is an \emph{interactive post-processing} of $\calF(\cdot; d)$-$FiltCon(IM)$.

\section{Concurrent Filter \& Odometer for $f$-DP}\label{sec.fDP}

$f$-DP is based on a hypothesis testing interpretation of differential privacy. Consider a hypothesis testing problem that attempts to quantify the difficulty of measuring the output of a mechanism on two adjacent datasets: 
\begin{center}
    $H_0$: the dataset is $x$ versus $H_1$: the dataset is $x'$
\end{center}
Let $Y := \calM(x)$ and $Y' := \calM(x')$ be the output distributions of mechanism $\calM$ on neighboring datasets $x, x'$. For a given rejection rule $\phi$, the type I error $\alpha_\phi = \mathbb E[\phi(Y)]$ is the probability of rejecting $H_0$ when $H_0$
is true, while the type II error $\beta_\phi = 1 - \mathbb E[\phi(Y')]$ is the probability of failing to reject $H_0$ when $H_1$ is true. 

The trade-off function characterizes the optimal boundary between achievable type I and type II errors, which in our case is calculated by finding the minimal achievable type II error after fixing the type I error at any level.

\begin{definition}[Trade-off function \cite{dong2019gaussian}]
    For any two probability distributions $Y$ and $Y'$ on the same space, the \emph{trade-off function} $T(Y, Y'): [0, 1] \to [0, 1]$ is defined as $$T(Y, Y')(\alpha) = \inf\{\beta_\phi: \alpha_\phi < \alpha\},$$ where the infimum is taken over all measurable rejection rules $\phi$.
\end{definition}

\begin{definition}[$f$-DP~\cite{dong2019gaussian}]\label{def.fdp}
    Let $f$ be a tradeoff function. A mechanism $\mathcal M: \mathcal X \to \mathbb R$ is $f$-differentially private if for every pair of neighboring datasets $x, x' \in \mathcal X$, we have $$T(\mathcal M(x), \mathcal M(x')) \geq f.$$
\end{definition}

In our generalized DP framework in Definition \ref{def.genprobdistance}, the probability distance measure for $f$-DP is $(S, \preceq, T)$. The distance mapping is the trade-off function $T$ in Definition \ref{def.fdp}. The partially ordered set $(S, \preceq)$ consists of the set $S$ of all trade-off functions $g: [0, 1] \to [0, 1]$ such that $g$ is convex, continuous, non-increasing, and $g(x) \leq 1 - x$ for $x \in [0, 1]$. The partial ordering is defined as $f_1\preceq f_2$ if  $f_1(\alpha)\ge f_2(\alpha)$ holds  for all $\alpha \in [0,1]$. A larger trade-off function means less privacy loss. 

The recent result of Vadhan and Zhang \cite{vadhan2022concurrent} shows that every interactive $f$-DP mechanism can be simulated by an interactive postprocessing of a noninteractive $f$-DP mechanism.

\begin{theorem}[Reduction of $f$-DP to noninteractive mechanism \cite{vadhan2022concurrent}]\label{thm.fdp_post}
	For every $f$ in the set $\mathcal S$ of trade-off functions, every pair of adjacent datasets $x, x'$, and every interactive $f$-DP mechanism $\calM$ with finite communication complexity, there exists a pair of random variables $Y, Y\,'$ with support $[0, 1]$ and a randomized interactive mechanism $\cP$ such that $D(Y, Y\,') \preceq f$.

    For every adversary $\cA$, we have 
    \begin{align*}
        \View(\cA \leftrightarrow \calM(x)) & \equiv \View(\cA \leftrightarrow \cP(Y)) \\
        \View(\cA \leftrightarrow \calM(x')) & \equiv \View(\cA \leftrightarrow \cP(Y\,')).
    \end{align*}
\end{theorem}

Essentially, the interactive postprocessing function in Theorem \ref{thm.fdp_post} takes the output of the non-interactive mechanism as the input, and interacts with the adversary just as the interactive mechanism does. It will use the output of the non-interactive mechanism to simulate all the answers to the adversary. 
This result is useful because it implies that for any fixed pair of datasets $x, x'$, to analyze the concurrent composition of interactive mechanisms $\mathcal M_i$, it suffices to consider the composition of their noninteractive versions $\mathcal N_i$. As a result, composition theorems for
noninteractive mechanisms extend to the concurrent composition of interactive $f$-DP mechanisms. 

Since $(\epsilon, \delta)$-DP is a specific case of $f$-DP \cite{dong2019gaussian}, where $f_{\epsilon, \delta} = \max \{ 0, 1 - \delta - \exp(\epsilon)\alpha, \exp(-\epsilon)(1 - \delta - \alpha)\}$, a corollary of this result is that an interactive $(\epsilon, \delta)$-DP mechanism can be postprocessed to a noninteractive $(\epsilon, \delta)$-DP mechanism, a result shown independently by Lyu \cite{lyu2022composition}. 

\subsection{Concurrent Filter for $f$-DP}

Using Theorem \ref{thm.fdp_post}, we derive a concurrent $f$-DP filter from an $f$-DP filter for noninteractive mechanisms. 

\begin{theorem}[Concurrent $f$-DP Filter]\label{thm.fDP.main}
Suppose that $\calF: \calD^* \times \calD \to \{1, 0\}$ is a valid 
$f$-DP continuation rule for noninteractive mechanisms. Then $\calF$ is a valid $f$-DP continuation rule for interactive mechanisms with finite communication.
\end{theorem}

\begin{proof}
We need to prove that for every $f \in \mathcal S$, $\calF(\cdot; f)$-$FiltCon(IM)$ is an $f$-DP I.M., which is equivalent to showing that for every pair of adjacent datasets $x, x'$, $\calF(\cdot; f)$-$FiltCon(IM)$ is $f$-DP on the dataset universe $\{x, x'\}$. Fixing $x, x'$, we will define an I.M.-to-I.M. postprocessing $\cP$ that takes $\NIF$ to $\CF$. 
For simplicity, we will assume that $\cA$ is deterministic, as Lemma \ref{lemma.deterministic_adversaries}  allows us to do. By the postprocessing property of $f$-DP IMs given in Theorem \ref{thm.fdp_post}, we know that an $f_j$-DP mechanism $\cM_i$ can be simulated by an interactive postprocessing $\cP_j$ of a noninteractive $f_j$-DP mechanism $\mathcal N_j$.

$\cP$ is constructed as follows. $\cP$ depends on whether it is primed to receive a query $v = Q$ from the adversary or an answer $v = A$ from the mechanism, to pass to the other party. If the adversary asks to start a new interactive mechanism, $\cP$ will make note of the corresponding noninteractive mechanism and the noninteractive-to-interactive postprocessing and pass the request on $\calF\text{-}FiltSeq(NIM)$. Otherwise, if the adversary queries an existing interactive mechanism $\calM_j$, $\cP$ will use the corresponding interactive postprocessing $\cP_j$ to answer the query instead of passing the message forward to $\cM$. The only time $\cP$ will interact with $\cM$ is to start a new mechanism and to pass back the confirmation that a new mechanism has begun to $\cA$. The algorithmic pseudocode of $\cP$ is in the appendix in Algorithm \ref{alg.fdp_nif_cf_post}.

Note that this algorithm only interacts with $\NIF$ in starting a new mechanism. For the rest of the interactive queries, the algorithm directly uses the interactive postprocessing of the corresponding noninteractive mechanism in question to answer the queries. By assumption, $\NIF$ is $f$-DP and $\CF \equiv \cP \circ \NIF$ on the dataset universe $\{x, x'\}$. Thus, by Corollary \ref{corollary.pim_privacy}, $\CF$ is $f$-DP.
\end{proof}

An analogue of Theorem \ref{thm.fDP.main} holds for $(\eps, \delta)$-DP. 
\begin{corollary}[Concurrent $(\eps, \delta)$ Filter] \label{corollary.eps_delta_f}
    Every filter for noninteractive $(\eps, \delta)$-DP mechanisms is also a concurrent filter of interactive $(\eps, \delta)$-DP mechanisms.
\end{corollary}

Concurrent $f$-DP odometers can be defined similarly from $f$-DP odometers for noninteractive mechanisms.

\begin{theorem}[Concurrent $f$-DP odometer]\label{thm.fDP.odometer}

Suppose that $\calG: \calD^* \to \calD$ is a valid 
$f$-DP privacy-loss accumulator for noninteractive mechanisms. Then $\calG$ is a valid $f$-DP privacy-loss accumulator for interactive mechanisms with finite communication.

\end{theorem}

\begin{proof}
    By Lemma \ref{lemma.f_odom}, if $\calG$ is a valid $f$-DP privacy-loss accumulator, then $\calF(\cdot; f) = I(\calG(\cdot) \preceq f)$ is a valid $f$-DP filter for noninteractive mechanisms. By Theorem \ref{thm.fDP.main}, if $\calF$ is an $f$-DP filter for noninteractive mechanisms, then $\calF$ will be an $f$-DP concurrent filter for interactive mechanisms. Applying Lemma \ref{lemma.f_odom} again, $\calG$ will be a valid $f$-DP privacy-loss accumulator for interactive mechanisms, and $\CO$ is an $f$-DP interactive mechanism. 
\end{proof}

Now we give an example of how to construct a valid $(\eps, \delta)$-DP privacy filter and odometer.

\begin{theorem}[$(\eps, \delta)$-DP privacy filter \cite{whitehouse2023fully}]\label{example.eps-delta-filter}
    For every $\delta' > 0$,  
    \begin{align*}
        & \calF((\eps_1, \delta_1), \dots, (\eps_k, \delta_k); (\eps, \delta)) \\ & = \mathbb I\left(\epsilon \leq \sqrt{2\log\left(\frac{1}{\delta'}\right)\sum_{m \leq k}\eps_m^2} + \frac{1}{2}\sum_{m \leq k} \eps_m^2\right) \cdot \mathbb I\left(\delta' + \sum_{m \leq k} \delta_k \leq \delta\right)
    \end{align*}
    is a valid approx-DP continuation rule for noninteractive mechanisms and interactive mechanisms.
\end{theorem}

\begin{theorem}[$(\eps, \delta)$-DP privacy odometer]\label{example.eps-delta-odom}
    Let $\delta = \delta' + \delta''$ be a target approximation parameter such that $\delta' > 0$, $\delta'' \geq 0$. Then 
    \begin{align*}
        & \calG((\eps_1, \delta_1), \dots, (\eps_k, \delta_k)) \\ 
        & = \begin{cases}
            \left(\sqrt{2\log\left(\frac{1}{\delta'}\right)\sum_{m \leq k}\eps_m^2} + \frac{1}{2}\sum_{m \leq k} \eps_m^2, \delta\right) & \text{if } \delta' +  \sum_{m \leq k} \delta_m \leq \delta \\
            (\infty, \infty) & \text{otherwise}
        \end{cases}
    \end{align*}
    is a valid approx-DP privacy-loss accumulator for noninteractive and interactive mechanisms. 
\end{theorem}

\section{Concurrent Filter \& Odometer for R\'enyi Divergence-Based DP}\label{sec.RDP}

R\'enyi DP is a relaxation of DP based on R\'enyi divergences.
In this section, we show that privacy loss adds up for R\'enyi DP, in the setting of concurrent composition with adaptive privacy-loss parameters. This bound matches naturally with the previous concurrent composition result where all privacy-loss parameters are fixed upfront. 

As a note, we inherit the limitation for when approximate-R\'enyi DP and zCDP are converted to approximate DP, in which there is some loss in the privacy parameters. This is an open question even for standard concurrent composition without adaptive privacy-loss parameters \cite{lyu2022composition, vadhan2022concurrent}, and we just inherit this limitation. The reason is that the current proof techniques for approximate DP and for R\'enyi DP are very different from each other, and it's not clear how to combine them. Hence, concurrent composition bounds for approximate RDP and zCDP is still an open question.

\begin{definition}[R\'enyi divergence  \cite{renyi1961measures}]
    For two probability distributions $P$ and $Q$, the \emph{R\'enyi divergence} of order $\alpha>1$ is 
	\begin{equation*}
		D_\alpha(P||Q)=\frac{1}{\alpha-1}\log \mathrm{E}_{x\sim Q} \left[\frac{P(x)}{Q(x)}\right]^\alpha.
	\end{equation*}
\end{definition}

\begin{definition}[R\'enyi DP \cite{mironov2017renyi}]
A randomized mechanism $\mathcal M$ is \emph{$(\alpha,\eps)$-R\'enyi differentially private} ($\eps$-RDP$_\alpha$)
if for all every two neighboring datasets $x$ and $x'$, $$D_{\alpha}(\mathcal M(x)||\mathcal M(x')) \leq \eps$$
\end{definition}

In our generalized DP framework in Definition \ref{def.genprobdistance}, for R\'enyi DP of order $\alpha$, the partially ordered set $\mathcal{D}$ is $((\mathbb{R}^{\ge0})\cup \{\infty\}, \le)$. The distance mapping is $\alpha$-R\'enyi divergence for $\alpha\in(1, \infty)$. R\'enyi divergence is also closed under postprocessing due to the data-processing inequality, and it satisfies the joint convexity. We will notate $\text{RDP}_\alpha$ as the family of distance measures corresponding to RDP, where $\calD = (\alpha, \cdot)\text{-RDP} = \text{RDP}_\alpha$.

Our proof leverages the $K=2$ case of the concurrent composition theorem proved by Lyu \cite{lyu2022composition}, where the privacy-loss parameters are specified unfront.

\begin{theorem}[Concurrent composition of R\'enyi-DP mechanisms \cite{lyu2022composition}]\label{thm.concomp_renyi}
    For all $\alpha > 1$, $k \in \mathbb N$, $\eps_1, \dots, \eps_k > 0$, and all interactive mechanisms $\calM_1, \dots, \calM_k$ such that $\calM_i$ is $\eps_i$-RDP$_\alpha$, the concurrent composition $\concomp(\calM_1, \dots, \calM_k)$ is $\sum_{i = 1}^k\eps_i$-RDP$_\alpha$. 
\end{theorem}

Another definition of differential privacy derived from R\'enyi divergence is zero-concentrated differential privacy.

\begin{definition}[Zero-concentrated differential privacy (zCDP) \cite{bun2016concentrated}]\label{def.zcdp}
    An interactive mechanism $\calM$ is \emph{$\rho$-zero-concentrated differentially private ($\rho$-zCDP)} if for every two neighboring datasets $x$ and $x'$, every adversary $\calA$, and all orders $\alpha \in (1, \infty)$, $$D_\alpha(\View(\calA \leftrightarrow \calM(x))||\View(\calA \leftrightarrow \calM(x'))) \leq \rho\alpha,$$
    where $D_\alpha$ is the R\'enyi divergence defined in Definition \ref{def.renyi_divergence}.

    In other words, $\calM$ is $\rho$-zCDP iff $\calM$ is $(\alpha, \rho\alpha)$-RDP for every $\alpha \in (1, \infty)$.
\end{definition}

Since zCDP bounds R\'enyi divergence for all orders $\alpha$, it is stronger than R\'enyi-DP of a fixed order. 

In our generalized probability framework in Definition \ref{def.genprobdistance}, the partially ordered set $\calD$ for zCDP is $((\mathbb R^{\geq 0} \cup \{ \infty \}, \leq)$. The distance mapping is the worst-case R\'enyi divergence normalized by $\alpha$, over all $\alpha \in (1, \infty)$. In other words, $D(\View(\calA \leftrightarrow \calM(x))||\View(\calA \leftrightarrow \calM(x'))) = \sup_{\alpha \in (1, \infty)} \left(\frac{D_\alpha(\View(\calA \leftrightarrow \calM(x))||\View(\calA \leftrightarrow \calM(x')))}{\alpha}\right)$. zCDP is closed under postprocessing \cite{bun2016concentrated, lyu2022composition} and satisfies the joint convexity property as well \cite{bun2016concentrated}. 

Similar to R\'enyi-DP, the concurrent composition of zCDP mechanisms is additive. 

\begin{theorem}[Concurrent composition of zCDP mechanisms \cite{lyu2022composition}]
    For $k \in \mathbb N$, $\rho_1, \dots, \rho_k > 0$, and all interactive mechanisms $\calM_1, \dots, \calM_k$ such that $\calM_i$ is $\rho_i$-zCDP, their concurrent composition $\concomp(\calM_1, \dots, \calM_k)$ is $\sum_{i = 1}^k \rho_i$-zCDP. 
\end{theorem}

We will see that once we have a concurrent R\'enyi-DP filter, the concurrent odometer case follows directly. Concurrent filters and odometers for zCDP will follow directly from the R\'enyi filter result as well. 

\subsection{Concurrent Filter \& Odometer for R\'enyi-DP}

\subsubsection{Filter}

\begin{theorem}[Theorem \ref{thm.main.RDP} restated]\label{thm.renyi-general-case}

An RDP filter with the continuation rule $\calF(\eps_1, \eps_2, \ldots;\eps)= \mathbb{I}(\sum_i \eps_i \le \eps)$ is a valid RDP IM-filter for every fixed order of $\alpha>1$.
\end{theorem}

Though generalized probability distance also covers RDP, the reduction in Theorem \ref{thm.fdp_post} holds only if the probability distance also satisfies the property of coupling\cite{vadhan2022concurrent}, which RDP does not. Therefore, we use a different strategy to prove \ref{thm.renyi-general-case}. We prove this theorem by inducting on a finite number of mechanisms being maintained by the $\calF$-filtered concurrent composition of interactive mechanisms $\CF$, and then taking the limit of those mechanisms. To do so, we first consider $\CF$ maintains up to $K$ mechanisms.

\begin{lemma}\label{lemma.renyi-general-case}
Let $\alpha > 1$ and define 
\begin{align*}
    \calF_K(\eps_1, \eps_2, \dots, \eps_k; \eps) = 
    \begin{cases}
        \mathbb I(\sum_{i = 1}^k \eps_i \leq \eps) & \text{ if $k \leq K$} \\
        0 & \text{ otherwise}
\end{cases}
\end{align*}
for every $k \in \mathbb N$. Then $\calF_K$ is a valid RDP$_\alpha$ IM-filter. 
\end{lemma}

\begin{proof}

We prove this by inducting on the maximum number $K$ of mechanisms that $\calF_K$-$FiltCon(IM)$ maintains. Fix a pair of datasets $x, x'$. By Lemma \ref{lemma.deterministic_adversaries}, it suffices to consider a deterministic adversary to prove this theorem.

\emph{Base case: $K = 2$.} Intuitively, with two mechanisms, the budget partition is determined before any mechanisms are queried, so we can reduce the problem to concurrent RDP composition. Suppose the adversary starts a $\eps_1$-RDP$_\alpha$ mechanism $\mathcal M_1$ through $\calF_2$-$FiltCon(IM)$. By the constraints of $\calF_2$-$FiltCon(IM)$, the second mechanism $\calM_2$ it is able to start must have privacy parameter $\eps_2 \leq \eps - \eps_1$. The adversary can adaptively interleave queries between $\mathcal M_1$ and $\mathcal M_2$. 

We will first define the notion of a \emph{universal} mechanism to help us with the reduction. Informally, a ``universal’’ mechanism can be any mechanism with the privacy loss bounded by the preset privacy loss budget. The universal mechanism will first verify if an $\mathcal{M}$ attempted by an adversary to start satisfies the given privacy loss parameter guarantee as well as if the given privacy loss parameter does not exceed the preset privacy loss budget. Upon verifying this, it will interact with the adversary as $\mathcal{M}$. Given a preset privacy budget $\eps$, an RDP$_\alpha$ universal mechanism $\mathcal U_{\alpha, \eps}(s, m)$ would operate as follows:

\SetKwFunction{AlgUniversalMech}{$\mathcal U_{\alpha, \eps}$}
\begin{algorithm}
    \DontPrintSemicolon
    \caption{Operation of $\mathcal U_{\alpha, \eps}(s, m)$.}
    
        \Fn{\AlgUniversalMech{$s, m$}}{
            \If{$m = \lambda$}{
                $x \gets s$, $s' \gets (\texttt{0}, x)$ \tcp*[f]{initialize state of $\mathcal U$}\;
                \Return{$(s, \lambda)$}\;
            }
            \If{$s$ is of form $s = (\texttt{0}, x)$ and $m = \calM$, where $\calM$ is an RDP$_\alpha$ mechanism}{
                \If{$\calM$ is $(\alpha, \eps_1)$-RDP}{
                    $(m', s_\calM) \gets \calM(s, \lambda)$, $s' \gets (\texttt{1}, s_\calM)$\;
                }\Else{
                    $m' \gets \texttt{invalid query}$, $s' \gets s$\;
                }
            }\ElseIf{$s$ is of form $s = (\texttt{1}, s_\calM)$}{
                $(s_\calM', m') \gets \calM(s_\calM, m)$\;
                $s' \gets (1, s_\calM')$\;
            }
            \Return{$(s', m')$}\;
        }
\end{algorithm}

The state of $\mathcal U_{\alpha, \eps}(s, m)$ only needs the state of $\calM$, the first mechanism it starts, and a bit to indicate where the first mechanism has started. 

Given a deterministic adversary $\mathcal A$ against $\calF_2$-$FiltCon(IM)$, let $(\calM_1, \eps_1) = \mathcal A(\lambda, \lambda)$, which is $\calA$'s first query to $\calF_2$-$FiltCon(IM)$. Assume that $\calM_1$ is $\eps_1$-RDP$_\alpha$ (else the query will be rejected by $\calF_2$-$FiltCon(IM)$, so we should use the first valid query of $\calA$). 

We will construct an adversary $\mathcal A'$ that interacts with $\concomp(\calM_1, \mathcal U_{\alpha, \eps - \eps_1})$, where $\mathcal U_{\alpha, \eps - \eps_1}$ is the universal $(\eps - \eps_1)$-RDP$_\alpha$ mechanism, and the adversary's first query to $\mathcal{U}_{\alpha, \eps - \eps_1}$ is an $\eps_2$-$\text{RDP}_\alpha$ mechanism $\mathcal{M}_2$. For each query $\mathcal A$ makes to $\calF_2$-$FiltCon(IM)$ prior to starting $\calM_2$, $\mathcal A'$ will make a query to $\calM_1$. When $\mathcal A$ starts $\calM_2$, $\mathcal A'$ will make its first query to $\mathcal U_{\alpha, \eps - \eps_1}$. Upon verifying $\eps_2\le \eps-\eps_1$, $\mathcal{U}_{\alpha, \eps - \eps_1}$ interacts with the adversary just as $\mathcal{M}_2$ does. Afterward, $\mathcal A'$ can interleave queries at will between $\calM_1$ and $\calM_2$, interacting with $\concomp(\calM_1, \mathcal U_{\alpha, \eps - \eps_1})$ just as $\mathcal A$ does with $\calF_2$-$FiltCon(IM)$.

In this case, $\View(\mathcal A' \leftrightarrow \concomp(\calM_1, \mathcal U_{\alpha, \eps - \eps_1})(x)) \equiv \View(\mathcal A' \leftrightarrow \cP \circ \concomp(\calM_1, \calM_2)(x))$ for an interactive postprocessing $\cP$ that does not pass on the adversary's first query to start a second mechanism to $\mathcal U_{\alpha, \eps - \eps_1}$, but otherwise preserves all of the adversary's queries and the mechanism's answers. The same holds for the computation on $x'$. By Theorem \ref{thm.concomp_renyi}, $\concomp(\mathcal M_1, \mathcal M_2)$ is $\eps$-RDP$_\alpha$, which means by Corollary \ref{corollary.pim_privacy} that $\concomp(\calM_1, \mathcal U_{\alpha, \eps - \eps_1})$ is as well. Therefore, $\calF_2$-$FiltCon(IM)$ is an $\eps$-RDP$_\alpha$ I.M. By Definition \ref{def.ddp_cf}, $\calF_2(\cdot; \eps)$ is a valid RDP$_\alpha$ IM-filter.

\emph{Induction on $K$.} By the inductive hypothesis, for every $\eps \geq 0$, $\calF_{K - 1}(\cdot; \eps)$-$FiltCon(IM)$ is $\eps$-RDP$_\alpha$. We define the postprocessing $\cP$ interacting with $\calF_2(\cdot; \eps)$-$FiltCon(IM)$ so that for any $K > 1$ and dataset $x$, $\View(\cA \leftrightarrow \cP \circ \calF_2(\cdot; \eps)$-$FiltCon(IM)(x)) \equiv \View(\cA \leftrightarrow \calF_{K}$-$FiltCon(IM)(x))$. The algorithmic pseudocode of $\cP$ is in the appendix in Algorithm \ref{alg.rdp_post}.

When receiving the adversary's request to start $\calM_1$ and verifying the privacy loss is within budget, $\cP$ will start both $\calM_1$ and $\calF_{K - 1}$-$FiltCon(IM)$ in $\calF_2$-$FiltCon(IM)$. $\cP$ will directly pass on queries to $\calM_1$. For queries $q$ to $\calM_j$ where $j > 1$ (in the form $(j, q)$), $\cP$ will postprocess the message as $(2, (j - 1), q)$, indicating that $\cP$ is passing $q$ on to the $(j - 1)$st mechanism of $\calF_{K - 1}$-$FiltCon(IM)$. (Concretely, suppose the adversary thinks it's interacting concurrently with $\calF_4$-$FiltCon(IM)$ maintaining $\cM_1, \cM_2, \cM_3, \cM_4$. In this construction, the adversary will actually be interacting with $\calF_2$-$FiltCon(IM)$ maintaining $\calM_1$ and $\calF_{3}$-$FiltCon(IM)$, the latter of which is comprised of $\cM_2, \cM_3, \cM_4$, which are the first, second, and third mechanisms, respectively, of $\calF_3$-$FiltCon(IM)$.)

Given a privacy-loss budget $\eps$, suppose the adversary starts a $\eps_1$-RDP$_\alpha$ mechanism $\mathcal M_1$ with $\eps_1 \le \eps$. As soon as $\calM_1$ gets created, $\cP$ will feed in $\calM_1$ to $\calF_2$-$FiltCon(IM)$. Upon receiving confirmation that $\calM_1$ has properly started, $\cP$ will then feed in $\calF_{K - 1}$ to $\calF_2$-$FiltCon(IM)$ with an $\epsilon - \epsilon_1$ budget. We don't need to handle the parsing in the $v = Q$ section because $\cP$ will always start two mechanisms at once. We can view $\calF_{K - 1}$-$FiltCon(IM)$ as the first query to $\mathcal U_2$.

By construction, for any deterministic adversary $\cA$ and dataset $x$, $\View(\cA \leftrightarrow \cP \circ \calF_{2}\text{-}FiltCon(IM)(x))$ is identically distributed as $ \View(\cA \leftrightarrow \calF_K$-$FiltCon(IM)(x))$. By Theorem \ref{thm.pim_privacy}, $\cP \circ \calF_{2}$-$FiltCon(IM)$ is an $\eps$-RDP$_\alpha$ I.M., which means $\calF_K\text{-}FiltCon(IM)$ is $\eps$-RDP$_\alpha$ I.M. By Definition \ref{def.ddp_cf}, $\calF_K$ is a valid RDP$_\alpha$ IM-filter.
\end{proof} 

To show that the filter works with an unbounded number of mechanisms, we examine the behavior as $K \to \infty$. For this to make sense, we observe that we are working with distributions on sequences that are well-behaved in that we can marginalize to finite subsequences. Suppose $(\Omega^\infty, \Gamma^\infty)$ is the direct product of an infinite sequence of measurable spaces $(\Omega_1, \Gamma_1), (\Omega_2, \Gamma_2), \dots$, i.e. $\Omega^\infty = \Omega_1 \times \Omega_2 \times \dots$ and $\Gamma^\infty$ is the smallest $\sigma$-algebra containing all cylinder sets $S_n(A) = \{ x^\infty \in \Omega^\infty \mid x_1, x_2, \dots, x_n \in A\}$, $A \in \Gamma^n$, for $n = 1, 2, \dots$ and $\Gamma^n = \Gamma_1 \otimes \Gamma_2 \otimes \dots \otimes \Gamma_n$. We call a sequence of probability distributions $P^1, P^2, \dots$ \emph{consistent} if $P^n$ is a distribution on $\Omega^n = \Omega_1 \times \Omega_2 \times \dots \times \Omega_n$ and $P^{i + 1}(A \times \Omega_{i + 1}) = P^i(A)$ for $A \in \Gamma^i$, $i = 1, \dots, n$.
For any such sequence there exists a distribution $P^\infty$ on $(\Omega^\infty, \Gamma^\infty)$ such that its marginal distribution on $\Omega^n$ is $P^n$, in the sense that $P^\infty(S_n(A)) = P^n(A)$, $A \in \Gamma^n$, where $\Omega^\infty = \Omega_1 \times \Omega_2 \times \dots $, $S_n(A)$ is the \emph{cylinder set} $\{ x^\infty \in \Omega^\infty \dots \} $, and $\Gamma^\infty$ is the smallest $\sigma$-algebra containing all the cylinder sets $S_k(\mathcal B)$ for $k \in \mathbb N$ and $\mathcal B \in \Gamma^k$.

\begin{lemma}[\cite{van2014renyi}]\label{lemma.infinite-divergence}
Let $P^1, P^2, \dots$ and $Q^1, Q^2, \dots$ be consistent sequences of probability distributions on $(\Omega^1, \Gamma^1), (\Omega^2, \Gamma^2), \dots$, where, for $n = 1, 2, \dots, \infty$, $(\Omega^n, \Gamma^n)$ is the direct product of the first $n$ measurable spaces in the infinite sequence $(\Omega_1, \Gamma_1), (\Omega_2, \Gamma_2), \dots$. Then for any $a \in (0, \infty]$ $$\lim_{n \to \infty} D_\alpha(P^n||Q^n) = D_\alpha(P^\infty||Q^\infty).$$ 
\end{lemma}

\begin{proof}[Proof of Theorem \ref{thm.renyi-general-case}.]

Fix a divergence parameter $\alpha\geq1$, and a pair of datasets $x, x'$. Consider $\eps \geq 0$, and initialize an adversary $\mathcal A$. Consider the views $V = (r_1, a_1, r_2, a_2, \dots)$ and $V' = (r_1', a_1', r_2', a_2', \dots)$ of $\calA$ interacting with $\calF(\cdot; \eps)$-$FiltCon(IM)$ on $x$ and $x'$, respectively. Let $q_i$ denote the adversary's $i$th query and $a_i$ denote the answer to the adversary's $i$th query.

Define sequences of random variables $V_1, V_2, V_3, \dots$ by $V_i = (r_1, a_1, r_2, a_2, \dots, r_i, a_i)$, and similarly for $V_1', V_2', V_3',\dots$.  By construction, both the distributions of $V_1, \dots$ and $V_1', \dots$ are consistent. 

By Lemma \ref{lemma.renyi-general-case}, for every $k = 1, 2, \dots$ and $\alpha > 1$, we have $D_\alpha(V_k||V_k') \leq \eps$. Then, invoking Lemma \ref{lemma.infinite-divergence}, we have $$D_\alpha(X_\infty||Y_\infty) =  \lim_{k \to \infty} D_\alpha(X_k||Y_k) \le \eps,$$
where the last inequality is because the sequence of $D_\alpha(X_k||Y_k)$ monotonically increases as $k \to \infty$, and is bounded above by $\eps$.
\end{proof}

Our approach can be used to provide another proof in \cite{feldman2022individual} for the bounds on privacy filters for sequentially-composed noninteractive DP mechanisms as well. They identify the sequence of losses incurred for each query -- in this case the invocation of a noninteractive mechanism and the reception of its answers -- as a supermartingale. They then apply the optional stopping theorem for martingales to bound the divergence between adjacent datasets and verify the validity of the filter. 

The basis for our proof is that concurrent composition of two interactive RDP mechanisms implies a fully adaptive $\calF_2$-$FiltCon(IM)$. Having interactivity allows us to proceed by induction in the proof of Lemma \ref{lemma.renyi-general-case}, describing $\calF_k$ in terms of $\calF_{k - 1}$. We could do the same for the sequential composition since induction preserves sequentiality. Then by Lemma \ref{lemma.renyi-general-case}, $\calF_k$-$FiltSeq(IM)$ is an $\eps$-RDP$_\alpha$ I.M. Thus, the sequential composition of two interactive R\'enyi DP mechanisms of order $\alpha$ implies a fully adaptive filter for sequential composition of RDP mechanisms of order $\alpha$. 

\subsubsection{Odometer}

Because of the bijective relationship between valid continuation rules and privacy-loss accumulators, we are able to define an additive relationship for concurrent odometers of R\'enyi mechanisms as well. 


\begin{theorem}[Valid $\eps$-RDP$_\alpha$ odometer for concurrently-composed interactive mechanisms]
    Let $\alpha > 1$. Define $\calG: \mathbb R_{\geq 0}^* \to \mathbb R_{\geq 0}$ by $\calG(\eps_1, \dots, \eps_K) = \sum_{i = 1}^K \eps_i$. $\calG$ is a valid privacy-loss accumulator for the concurrent composition of a sequence of mechanisms $\calM_1, \calM_2, \dots,$ where $\calM_i$ is an $\eps_i$-RDP$_\alpha$ interactive mechanism.
\end{theorem}
    
\begin{proof}
Consider the relation $\calF(\eps_1, \dots, \eps_K; \eps): \calD^* \times \calD \to \{ 0, 1\}$ such that $\calF(\eps_1, \dots, \eps_K; \eps) = \mathbb I(\calG(\eps_1, \dots, \eps_K \leq \eps)$. By Theorem \ref{thm.renyi-general-case}, $\calF$ is a valid RDP$_\alpha$ filter. By Lemma \ref{lemma.f_odom}, $\calG$ is a valid privacy-loss accumulator for R\'enyi mechanisms of order $\alpha$.
\end{proof}

\subsection{Concurrent Filter \& Odometer for zCDP}

\begin{theorem}[Valid zCDP filter for concurrently-composed interactive mechanisms]\label{thm.zcdp-filter}

Define the relation $\calF: \mathbb R_{\geq 0}^* \times \mathbb R_{\geq 0} \to \{ 0, 1\}$ by  $\calF(\rho_1, \rho_2, \dots; \rho)= \mathbb{I}(\sum_i \rho_i \le \rho)$. Then $\calF$ is a valid $\rho$-zCDP filter for the concurrent composition of zCDP interactive mechanisms. 
\end{theorem}

\begin{proof}
    By Definition \ref{def.zcdp}, a mechanism $\calM_i$ that is $\rho_i$-zCDP is also $(\alpha, \rho_i\alpha)$-RDP for all $\alpha > 1$. For a particular $\alpha$, by Theorem \ref{thm.renyi-general-case}, $\mathbb I(\sum_{i} \rho_i\alpha \leq \rho\alpha) = \mathbb I(\alpha\sum_{i} \rho_i \leq \rho\alpha) = \mathbb I(\sum_{i} \rho_i \leq \rho)$ is a valid RDP IM-filter, meaning that a filter $\calF^{(\alpha)}$-$FiltCon(IM)$ constructed from $\calF^{(\alpha)}(\rho_1\alpha, \rho_2\alpha, \dots; \rho\alpha) := \mathbb I(\sum_i \rho_i\alpha \leq \rho\alpha) = \calF(\rho_1, \rho_2, \dots; \rho)$ is an $(\alpha, \rho\alpha)$-RDP I.M. Since Theorem \ref{thm.renyi-general-case} holds for all $\alpha$, $\calF^{(\alpha)}$-$FiltCon(IM)$ is an $(\alpha, \rho\alpha)$-RDP I.M. for all $\alpha > 1$. 
    
    To show that $\calF(\rho_1, \rho_2, \dots; \rho) = \mathbb I(\sum_i \rho_i \leq \rho)$ is a valid $\rho$-zCDP IM-filter, we construct a postprocessing $\cP$ that transforms $\calF^{(\alpha)}(\cdot; \rho\alpha)$-$FiltCon(IM)$ to $\calF(\cdot; \rho)$-$FiltCon(IM)$. Informally, $\cP$ converts all queries from the adversary that start $\rho_i$-zCDP mechanisms into $(\alpha, \rho_i\alpha)$-RDP mechanisms, and preserves all other messages. The algorithmic pseudocode of $\cP$ is in the appendix in Algorithm \ref{alg.zcdp_post}. 

    By Corollary \ref{corollary.pim_privacy}, because $\View(\calA \leftrightarrow \calF$-$FiltCon(IM)) \equiv \View(\calA \leftrightarrow \cP \circ \calF^{(\alpha)}$-$FiltCon(IM))$, $\calF(\cdot; \rho)$-$FiltCon(IM)$ is a valid $\rho$-zCDP I.M., which means $\calF(\cdot; \rho)$ is a valid $\rho$-zCDP IM-filter.
\end{proof}

\begin{theorem}[Valid $\rho$-zCDP odometer for concurrently-composed interactive mechanisms]
    Define $\calG: \calD^* \to \calD$ to be $\calG(\rho_1, \rho_2, \dots) = \sum_{i = 1} \rho_i$. Then $\calG$ is a valid privacy-loss accumulator for the concurrent composition of a sequence of mechanisms $\calM_1, \calM_2, \dots,$ where $\calM_i$ is an $\rho_i$-zCDP interactive mechanism. 
\end{theorem}
    
\begin{proof}
Define a relation $\calF(\rho_1, \rho_2, \dots; \rho): \calD^* \times \calD \to \calD$ such that $\calF(\rho_1, \rho_2, \dots; \rho) = \mathbb I(\calG(\rho_1, \rho_2, \dots \leq \rho)$. By Theorem \ref{thm.zcdp-filter}, $\calF$ is a valid $\rho$-zCDP filter. By Lemma \ref{lemma.f_odom}, $\calG$ is a valid zCDP privacy-loss accumulator. 
\end{proof}
\section{Implementation}\label{sec.implement}

In this section, we provide an overview of the implementation of differentially private mechanisms in the open-source software project OpenDP and the Tumult Analytics platform. We then discuss the implications of our results on the two platforms. 

\subsection{Practical Application in the OpenDP Library}
The OpenDP Library \cite{gaboardi2020programming} is a modular collection of statistical algorithms used to build differentially private computations. 
The OpenDP Library represents differentially private computations with measurements and odometers. 
Measurements and odometers fully characterize the privacy guarantees they give in terms of an input domain, input metric, and privacy measure. 
They both contain a function to make a differentially private release and a privacy map to reason about the privacy spend of a release.
A privacy map is a function that takes in a bound on the distance between adjacent inputs/datasets ($d_{in}$),
and translates it to a bound on the distance between respective outputs ($d_{out}$, a privacy parameter).
From the perspective of the library, an \emph{interactive measurement} is simply a measurement for which the function emits a \emph{queryable} (Figure \ref{fig:queryable}), OpenDP terminology for an object that implements an interactive mechanism. 
A queryable is modeled as an opaque state machine, consisting of an internal state value and a transition function.
When an analyst passes a query into the queryable, the state is updated, and an answer is returned.

\begin{figure}
    \centering
    \includegraphics[width=\linewidth]{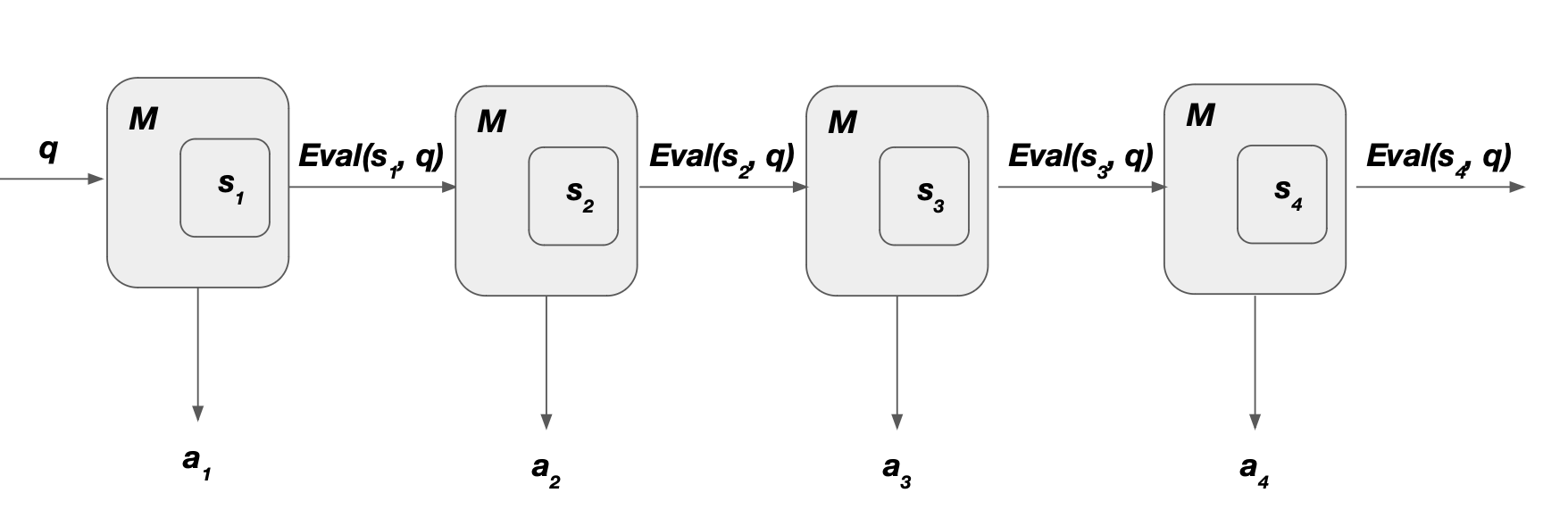}
    \caption{Diagram of a queryable.}
    \label{fig:queryable}
\end{figure}

All interactivity in the OpenDP Library uses queryables, and the privacy guarantee of an interactive measurement is defined exactly as in Definitions \ref{def.im_e_d} and \ref{def.ddp}, quantitatively for all possible adversary's strategies $\calA$.

Interactive measurements are well-suited to capture interactive composition when the desired privacy parameters for each query are known up-front. 
For example, the internal state of a compositor queryable consists of the dataset and per-query privacy parameters.
Each time a query $q_i$ (which may itself be a measurement) is passed to the compositor queryable, 
a privacy parameter (like $\eps_i$) is popped off of its internal state.
Unlike measurements, odometers always emit a queryable from their function, 
and the privacy map is maintained \emph{inside} the queryable.
OpenDP's implementation of odometers differs from the pseudocode in this paper, 
in that the input distance is not fixed up-front:
the \texttt{privacy\_loss} query also holds a $d_{in}$.
Thus, queryables spawned by odometers can be queried with a bound on the distance between adjacent input datasets $d_{in}$, 
which returns the current corresponding privacy spend $d_{out}$.

The state of odometer queryables in OpenDP consist of the dataset and a vector of privacy maps (one from each query).
When the odometer queryable is passed a privacy-loss query, 
the queryable passes $d_{in}$ from the query into each of the child privacy maps stored in its state,
which emit the $d_i$ used in Algorithms \ref{alg.o_construction} and \ref{alg.co_construction}.
The privacy loss returned is the composition of these privacy loss parameters.

A filter is an instance of an interactive measurement.
Any odometer can be converted into a filter IM by fixing an upper bound on the privacy parameters, as in Lemma \ref{lemma.f_odom}.
That is, any queryable spawned by the filter IM will refuse to answer any queries that would exceed the configured privacy parameters.

The OpenDP Library uses ``wrapping'' to handle situations where a queryable must influence the behavior of queryables it spawns in order to uphold its privacy guarantees.
For example, a sequential composition queryable must wrap any queryables it spawns in a logic that first asks the sequential compositor queryable for permission to execute (which allows the sequential compositor to maintain sequentiality), executes the query, and then recursively wraps any queryable in the answer with the same logic.

Concurrent composition improves upon sequential composition in that concurrent composition does not need to influence the behavior of child queryables.
Thus, concurrent compositors, odometers, and filters benefit from a simple implementation, in that they do not need to implement complicated sequentiality constraints via ``wrapping''.
More specifically, implementing \emph{any} composition primitive for interactive mechanisms is more complex than for noninteractive mechanisms, as the composition primitive needs to provide the analyst/adversary access to many interactive mechanisms (in such a way that the only access to the sensitive data is through queries, and in particular ensuring that the private internal state of the mechanisms cannot be inspected directly), in contrast to simply sending the analyst the results of noninteractive mechanisms. In the case of sequential composition, the compositor implementation had to be even more complex in order to enforce sequentiality; our work removes the need for this additional complexity.

Concurrent composition also allows the analyst to work in a more exploratory fashion. 
By allowing non-sequential access to mechanisms, the analyst is free to sequence their queries arbitrarily.
Analysts may also return to previous (un-exhausted) mechanisms without incurring additional privacy penalties for burn-in.
This will also allow analysts to interact with multiple mechanisms simultaneously in settings where computational concurrency may make sequencing ambiguous.
Our filter and odometer theorems also allow analysts to choose privacy parameters adaptively, and partition their privacy budget across different mechanisms, providing added flexibility. 

We provide our implementation of concurrent composition based on the OpenDP Library on Github at \url{https://github.com/concurrent-composition/concurrent-composition}.

\subsection{Practical Application in the Tumult Framework}

Tumult Analytics \cite{tumultanalyticssoftware} \cite{tumultanalyticswhitepaper} is an open-source framework for releasing aggregate information from sensitive datasets with differential privacy. Analytics queries are evaluated using differentially private mechanisms in the context of a Tumult Analytics \emph{session}, which is essentially a privacy filter. Analytics is written on top of Tumult Core, a framework for constructing differentially private mechanisms with automated privacy guarantees. Interactive mechanisms are the basic abstraction for interactivity in Tumult Core (an interactive mechanism is called a \emph{queryable} in Tumult Core). Tumult Core currently restricts interleaving queries to the composition of queryables to prevent concurrent composition. For example, suppose we start with a privacy filter queryable, and our first query spawns a new queryable using some portion of the budget. Tumult Core requires that the user finish interacting with this spawned queryable before asking a new query to the parent privacy filter, otherwise we may (adaptively) spawn a second queryable and concurrently ask queries of both spawned queryables. Note that while certain concurrent composition results are known from prior work, Tumult Core currently does not allow concurrent composition under any circumstances to reduce code complexity and maintain UX consistency. The results in this paper allows for removing this restriction.

Although spawning queryables within queryables may seem niche, one common application in Tumult Analytics is with parallel composition. Parallel composition in differential privacy says that if we partition a dataset and run DP queries on each subset in the partition, the overall privacy loss is the maximum privacy loss across the mechanisms. In Tumult Analytics, the mechanisms the user runs on each subset of the partitioned data can be interactive. Under the hood, this is implemented as follows: the top level privacy filter spawns a new interactive mechanism, consisting of multiple privacy filters (one for each subset in the partition). To prevent concurrent composition where it is not allowed, we require that the user finish interacting with this parallel composition interactive mechanism before they can ask any more queries of the top-level privacy filter, or spawn a new parallel composition mechanism from the top-level privacy filter.

The results in this paper allow for removing the restrictions on concurrent composition, and thereby improve the existing system in two ways. First, it reduces code complexity and improves auditablility. The privacy guarantee of an interactive mechanism is contingent on the promise that there will be no concurrent composition in scenarios where it is not allowed. This property is more challenging to verify than other properties of the system that contribute to the privacy guarantee required, which tend to be localized and easy to verify in isolation.

Second, removing restrictions on concurrent composition improves the UX of Tumult Analytics. As mentioned previously, users of Analytics must finish interacting with one queryable before they start interacting with a subsequent queryable, and this multi-queryable scenario happens frequently when using parallel composition. This has the potential to cause issues:
\begin{itemize}
    \item It is a possible a user could have a reason for wanting to interleave queries to two or more queryables.
    \item Even if interleaving queries isn't required, the user may structure their code such that queries to multiple queryables are interleaved. This would cause the user to see an unfamiliar error, and require them to restructure their code.
\end{itemize}
Both of these issues would be solved if restrictions on concurrent composition for privacy filters was removed, which the results of our paper can enable.

\section*{Acknowledgement}
S.V. is supported by a gift from Apple, NSF grant BCS-2218803, a grant from the Sloan Foundation and a Simons Investigator Award. V.X. is supported in part by NSF grant BCS-2218803. W.Z. is supported in part by a Computing Innovation Fellowship from the Computing Research Association (CRA) and the Computing Community Consortium (CCC), and the BU Census grant. 

We thank the JPC referees for their detailed reviews and valuable feedback. We are especially grateful to Roodabeh Safavi Hemami for identifying an error in Algorithm 7 in a previous version of this work.

\bibliographystyle{ACM-Reference-Format}
\balance
\bibliography{JPC/ref-final}


\begin{thebibliography}{24}


\ifx \showCODEN    \undefined \def \showCODEN     #1{\unskip}     \fi
\ifx \showDOI      \undefined \def \showDOI       #1{#1}\fi
\ifx \showISBNx    \undefined \def \showISBNx     #1{\unskip}     \fi
\ifx \showISBNxiii \undefined \def \showISBNxiii  #1{\unskip}     \fi
\ifx \showISSN     \undefined \def \showISSN      #1{\unskip}     \fi
\ifx \showLCCN     \undefined \def \showLCCN      #1{\unskip}     \fi
\ifx \shownote     \undefined \def \shownote      #1{#1}          \fi
\ifx \showarticletitle \undefined \def \showarticletitle #1{#1}   \fi
\ifx \showURL      \undefined \def \showURL       {\relax}        \fi
\providecommand\bibfield[2]{#2}
\providecommand\bibinfo[2]{#2}
\providecommand\natexlab[1]{#1}
\providecommand\showeprint[2][]{arXiv:#2}

\bibitem[Berghel et~al\mbox{.}(2022)]%
        {tumultanalyticswhitepaper}
\bibfield{author}{\bibinfo{person}{Skye Berghel}, \bibinfo{person}{Philip Bohannon}, \bibinfo{person}{Damien Desfontaines}, \bibinfo{person}{Charles Estes}, \bibinfo{person}{Sam Haney}, \bibinfo{person}{Luke Hartman}, \bibinfo{person}{Michael Hay}, \bibinfo{person}{Ashwin Machanavajjhala}, \bibinfo{person}{Tom Magerlein}, \bibinfo{person}{Gerome Miklau}, \bibinfo{person}{Amritha Pai}, \bibinfo{person}{William Sexton}, {and} \bibinfo{person}{Ruchit Shrestha}.} \bibinfo{year}{2022}\natexlab{}.
\newblock \showarticletitle{Tumult {{Analytics}}: a robust, easy-to-use, scalable, and expressive framework for differential privacy}.
\newblock \bibinfo{journal}{\emph{arXiv preprint arXiv:2212.04133}} (\bibinfo{date}{Dec.} \bibinfo{year}{2022}).
\newblock


\bibitem[Bun and Steinke(2016)]%
        {bun2016concentrated}
\bibfield{author}{\bibinfo{person}{Mark Bun} {and} \bibinfo{person}{Thomas Steinke}.} \bibinfo{year}{2016}\natexlab{}.
\newblock \showarticletitle{Concentrated differential privacy: Simplifications, extensions, and lower bounds}. In \bibinfo{booktitle}{\emph{Theory of Cryptography Conference}}. Springer, \bibinfo{pages}{635--658}.
\newblock


\bibitem[Dong et~al\mbox{.}(2019)]%
        {dong2019gaussian}
\bibfield{author}{\bibinfo{person}{Jinshuo Dong}, \bibinfo{person}{Aaron Roth}, {and} \bibinfo{person}{Weijie~J Su}.} \bibinfo{year}{2019}\natexlab{}.
\newblock \showarticletitle{Gaussian differential privacy}.
\newblock \bibinfo{journal}{\emph{arXiv preprint arXiv:1905.02383}} (\bibinfo{year}{2019}).
\newblock


\bibitem[Dwork et~al\mbox{.}(2006)]%
        {dwork2006our}
\bibfield{author}{\bibinfo{person}{Cynthia Dwork}, \bibinfo{person}{Krishnaram Kenthapadi}, \bibinfo{person}{Frank McSherry}, \bibinfo{person}{Ilya Mironov}, {and} \bibinfo{person}{Moni Naor}.} \bibinfo{year}{2006}\natexlab{}.
\newblock \showarticletitle{Our data, ourselves: Privacy via distributed noise generation}. In \bibinfo{booktitle}{\emph{Annual international conference on the theory and applications of cryptographic techniques}}. Springer, \bibinfo{pages}{486--503}.
\newblock


\bibitem[Dwork et~al\mbox{.}(2010a)]%
        {DNPR10}
\bibfield{author}{\bibinfo{person}{Cynthia Dwork}, \bibinfo{person}{Moni Naor}, \bibinfo{person}{Toniann Pitassi}, {and} \bibinfo{person}{Guy~N. Rothblum}.} \bibinfo{year}{2010}\natexlab{a}.
\newblock \showarticletitle{Differential privacy under continual observation}. In \bibinfo{booktitle}{\emph{Proceedings of the 42nd ACM Symposium on Theory of Computing}} \emph{(\bibinfo{series}{STOC '10})}. \bibinfo{pages}{715--724}.
\newblock


\bibitem[Dwork et~al\mbox{.}(2009)]%
        {DNRRV09}
\bibfield{author}{\bibinfo{person}{Cynthia Dwork}, \bibinfo{person}{Moni Naor}, \bibinfo{person}{Omer Reingold}, \bibinfo{person}{Guy~N. Rothblum}, {and} \bibinfo{person}{Salil~P. Vadhan}.} \bibinfo{year}{2009}\natexlab{}.
\newblock \showarticletitle{On the complexity of differentially private data release: efficient algorithms and hardness results}. In \bibinfo{booktitle}{\emph{Proceedings of the 41st ACM Symposium on Theory of Computing}} \emph{(\bibinfo{series}{STOC '09})}. \bibinfo{pages}{381--390}.
\newblock


\bibitem[Dwork and Roth(2014)]%
        {dwork2014algorithmic}
\bibfield{author}{\bibinfo{person}{Cynthia Dwork} {and} \bibinfo{person}{Aaron Roth}.} \bibinfo{year}{2014}\natexlab{}.
\newblock \showarticletitle{The algorithmic foundations of differential privacy}.
\newblock \bibinfo{journal}{\emph{Foundations and Trends in Theoretical Computer Science}} \bibinfo{volume}{9}, \bibinfo{number}{3--4} (\bibinfo{year}{2014}), \bibinfo{pages}{211--407}.
\newblock


\bibitem[Dwork and Rothblum(2016)]%
        {dwork2016concentrated}
\bibfield{author}{\bibinfo{person}{Cynthia Dwork} {and} \bibinfo{person}{Guy~N Rothblum}.} \bibinfo{year}{2016}\natexlab{}.
\newblock \showarticletitle{Concentrated differential privacy}.
\newblock \bibinfo{journal}{\emph{arXiv preprint arXiv:1603.01887}} (\bibinfo{year}{2016}).
\newblock


\bibitem[Dwork et~al\mbox{.}(2010b)]%
        {dwork2010boosting}
\bibfield{author}{\bibinfo{person}{Cynthia Dwork}, \bibinfo{person}{Guy~N Rothblum}, {and} \bibinfo{person}{Salil Vadhan}.} \bibinfo{year}{2010}\natexlab{b}.
\newblock \showarticletitle{Boosting and differential privacy}. In \bibinfo{booktitle}{\emph{2010 IEEE 51st Annual Symposium on Foundations of Computer Science}}. IEEE, \bibinfo{pages}{51--60}.
\newblock


\bibitem[Feldman and Zrnic(2022)]%
        {feldman2022individual}
\bibfield{author}{\bibinfo{person}{Vitaly Feldman} {and} \bibinfo{person}{Tijana Zrnic}.} \bibinfo{year}{2022}\natexlab{}.
\newblock \showarticletitle{Individual Privacy Accounting via a Renyi Filter}.
\newblock \bibinfo{journal}{\emph{arXiv preprint arXiv:2008.11193}} (\bibinfo{year}{2022}).
\newblock


\bibitem[Gaboardi et~al\mbox{.}(2020)]%
        {gaboardi2020programming}
\bibfield{author}{\bibinfo{person}{Marco Gaboardi}, \bibinfo{person}{Michael Hay}, {and} \bibinfo{person}{Salil Vadhan}.} \bibinfo{year}{2020}\natexlab{}.
\newblock \showarticletitle{A programming framework for opendp}.
\newblock \bibinfo{journal}{\emph{Manuscript}} (\bibinfo{year}{2020}).
\newblock


\bibitem[Hardt and Rothblum(2010)]%
        {hardt2010multiplicative}
\bibfield{author}{\bibinfo{person}{Moritz Hardt} {and} \bibinfo{person}{Guy~N Rothblum}.} \bibinfo{year}{2010}\natexlab{}.
\newblock \showarticletitle{A multiplicative weights mechanism for privacy-preserving data analysis}. In \bibinfo{booktitle}{\emph{2010 IEEE 51st annual symposium on foundations of computer science}}. IEEE, \bibinfo{pages}{61--70}.
\newblock


\bibitem[Kairouz et~al\mbox{.}(2015)]%
        {kairouz2015composition}
\bibfield{author}{\bibinfo{person}{Peter Kairouz}, \bibinfo{person}{Sewoong Oh}, {and} \bibinfo{person}{Pramod Viswanath}.} \bibinfo{year}{2015}\natexlab{}.
\newblock \showarticletitle{The composition theorem for differential privacy}. In \bibinfo{booktitle}{\emph{International conference on machine learning}}. PMLR, \bibinfo{pages}{1376--1385}.
\newblock


\bibitem[Labs(2022)]%
        {tumultanalyticssoftware}
\bibfield{author}{\bibinfo{person}{Tumult Labs}.} \bibinfo{year}{2022}\natexlab{}.
\newblock \bibinfo{booktitle}{\emph{Tumult {{Analytics}}}}.
\newblock
\urldef\tempurl%
\url{https://tmlt.dev}
\showURL{%
\tempurl}


\bibitem[L{\'e}cuyer(2021)]%
        {lecuyer2021practical}
\bibfield{author}{\bibinfo{person}{Mathias L{\'e}cuyer}.} \bibinfo{year}{2021}\natexlab{}.
\newblock \showarticletitle{Practical Privacy Filters and Odometers with R$\backslash$'enyi Differential Privacy and Applications to Differentially Private Deep Learning}.
\newblock \bibinfo{journal}{\emph{arXiv preprint arXiv:2103.01379}} (\bibinfo{year}{2021}).
\newblock


\bibitem[Lyu(2022)]%
        {lyu2022composition}
\bibfield{author}{\bibinfo{person}{Xin Lyu}.} \bibinfo{year}{2022}\natexlab{}.
\newblock \showarticletitle{Composition Theorems for Interactive Differential Privacy}. In \bibinfo{booktitle}{\emph{Thirty-sixth Conference on Neural Information Processing Systems}}.
\newblock


\bibitem[Mironov(2017)]%
        {mironov2017renyi}
\bibfield{author}{\bibinfo{person}{Ilya Mironov}.} \bibinfo{year}{2017}\natexlab{}.
\newblock \showarticletitle{R{\'e}nyi differential privacy}. In \bibinfo{booktitle}{\emph{2017 IEEE 30th Computer Security Foundations Symposium (CSF)}}. IEEE, \bibinfo{pages}{263--275}.
\newblock


\bibitem[Murtagh and Vadhan(2016)]%
        {murtagh2016complexity}
\bibfield{author}{\bibinfo{person}{Jack Murtagh} {and} \bibinfo{person}{Salil Vadhan}.} \bibinfo{year}{2016}\natexlab{}.
\newblock \showarticletitle{The complexity of computing the optimal composition of differential privacy}. In \bibinfo{booktitle}{\emph{Theory of Cryptography Conference}}. Springer, \bibinfo{pages}{157--175}.
\newblock


\bibitem[R{\'e}nyi(1961)]%
        {renyi1961measures}
\bibfield{author}{\bibinfo{person}{Alfr{\'e}d R{\'e}nyi}.} \bibinfo{year}{1961}\natexlab{}.
\newblock \showarticletitle{On measures of entropy and information}. In \bibinfo{booktitle}{\emph{Proceedings of the fourth Berkeley symposium on mathematical statistics and probability}}, Vol.~\bibinfo{volume}{1}. Berkeley, California, USA.
\newblock


\bibitem[Rogers et~al\mbox{.}(2021)]%
        {rogers2021privacy}
\bibfield{author}{\bibinfo{person}{Ryan Rogers}, \bibinfo{person}{Aaron Roth}, \bibinfo{person}{Jonathan Ullman}, {and} \bibinfo{person}{Salil Vadhan}.} \bibinfo{year}{2021}\natexlab{}.
\newblock \bibinfo{title}{Privacy Odometers and Filters: Pay-as-you-Go Composition}.
\newblock
\newblock
\showeprint[arxiv]{1605.08294}~[cs.CR]


\bibitem[Vadhan and Wang(2021)]%
        {vadhan2021concurrent}
\bibfield{author}{\bibinfo{person}{Salil Vadhan} {and} \bibinfo{person}{Tianhao Wang}.} \bibinfo{year}{2021}\natexlab{}.
\newblock \showarticletitle{Concurrent Composition of Differential Privacy}. In \bibinfo{booktitle}{\emph{Theory of Cryptography Conference}}. Springer, \bibinfo{pages}{582--604}.
\newblock


\bibitem[Vadhan and Zhang(2023)]%
        {vadhan2022concurrent}
\bibfield{author}{\bibinfo{person}{Salil Vadhan} {and} \bibinfo{person}{Wanrong Zhang}.} \bibinfo{year}{2023}\natexlab{}.
\newblock \showarticletitle{Concurrent Composition Theorems for Differential Privacy}. In \bibinfo{booktitle}{\emph{55th Annual ACM Symposium on Theory of Computing}}.
\newblock


\bibitem[Van~Erven and Harremos(2014)]%
        {van2014renyi}
\bibfield{author}{\bibinfo{person}{Tim Van~Erven} {and} \bibinfo{person}{Peter Harremos}.} \bibinfo{year}{2014}\natexlab{}.
\newblock \showarticletitle{R{\'e}nyi divergence and Kullback-Leibler divergence}.
\newblock \bibinfo{journal}{\emph{IEEE Transactions on Information Theory}} \bibinfo{volume}{60}, \bibinfo{number}{7} (\bibinfo{year}{2014}), \bibinfo{pages}{3797--3820}.
\newblock


\bibitem[Whitehouse et~al\mbox{.}(2023)]%
        {whitehouse2023fully}
\bibfield{author}{\bibinfo{person}{Justin Whitehouse}, \bibinfo{person}{Aaditya Ramdas}, \bibinfo{person}{Ryan Rogers}, {and} \bibinfo{person}{Zhiwei~Steven Wu}.} \bibinfo{year}{2023}\natexlab{}.
\newblock \bibinfo{title}{Fully Adaptive Composition in Differential Privacy}.
\newblock
\newblock
\showeprint[arxiv]{2203.05481}~[cs.LG]


\end{thebibliography}

\newpage
\appendix

\section{Algorithms}
\subsection{Algorithm for proof of Theorem \ref{thm.fDP.main}}

In this section, we enclose the IM-to-IM postprocessing algorithm $\cP$ for the proof of Theorem \ref{thm.fDP.main}, as defined in Algorithm \ref{alg.fdp_nif_cf_post}. $\cP$ will interact with $\NIF$ only to initialize and return $\NIF$'s response for a new noninteractive $f$-DP mechanism, if an adversary asks to start an interactive $f$-DP mechanism.  Otherwise, for a query to mechanism $\calM_j$, $\cP$ will use the NIM-to-IM postprocessing of $\calM_j$, $\cP_j$, to answer the query. This allows any adversary $\calA$ to interact with $\NIF$ as if it were interacting with $\CF$.

\SetKwFunction{AlgPostprocessing}{$\cP$}

\begin{algorithm}
    \DontPrintSemicolon
    \caption{Interactive postprocessing $\cP$ that transforms $f$-DP $\NIF$ into $\CF$.}\label{alg.fdp_nif_cf_post}
        \Fn{\AlgPostprocessing{$s, v, m$}}{
            \If(\tcp*[f]{initialize $\cP$}){$s = \lambda$}{
                $s \gets ([])$, where $[]$ is an empty list\;
            }
            \If(\tcp*[f]{$\cP$ accepting queries}){$v = Q$}{
                Parse $s$ as $([(\cP_1, s^{\cP_1}, f_1), \dots, (\cP_k, s^{\cP_k}, f_k)])$\;
                \If(\tcp*[f]{$\cA$ starts a new mechanism}){$m$ is of the form $m = (\cM', f')$}{
                    Let $\cP_{k + 1}$ be the interactive postprocessing for noninteractive mechanism $\mathcal N_{k + 1}$ so that $\cP_{k + 1}\circ \mathcal N_{k + 1} \equiv \cM_{k + 1}$ on dataset universe $\{x, x'\}$\;
                    $s' \gets ([(\cP_1, s^{\cP_1}, f_1), \dots, (\cP_{k + 1}, \texttt{pending}, f_{k + 1})])$\;
                    $m' \gets (\mathcal N_{k + 1}, f_{k + 1})$\;
                }\ElseIf(\tcp*[f]{$\cA$ queries an old mechanism}){$m$ is of form $m = (j, q)$ where $j = 1, \dots, k$ and $q$ is a query to $\cM_j$}{
                    $(s', m') \gets \cP_j(s, q)$\;
                    $s' \gets ([(\cP_1, s^{\cP_1}, f_1), \dots, (\cP_j, (s^{\cP_j})', f_j), \dots (\cP_k, s^{\cP_k}, f_k)])$\;
                }
            }\ElseIf(\tcp*[f]{$\cP$ accepting answers}){$v = A$}{
                \If{$s$ is of form $((\cP_1, s^{\cP_1}, f_1), \dots, (\cP_k, \texttt{pending}, f_k))$}{
                    \If(\tcp*[f]{did not successfully start a new mechanism}){$m = \texttt{invalid query}$}{
                        $s' \gets ([(\cP_1, s^{\cP_1}, f_1), \dots, (\cP_{k - 1}, s^{\cP_{k - 1}}, f_{k - 1})])$\;
                        $m' \gets m$\;
                    }\Else(\tcp*[f]{noninteractive output of new mechanism}){
                        $(s^{\cP_k}, m') \gets \cP(\lambda, A, m)$ \tcp*[f]{set up new interactive postprocessor's state}\;
                        $s' \gets ([(\cP_1, s^{\cP_1}, f_1), \dots, (\cP_k, s^{\cP_k}, f_k)])$\;
                    }
                }
            }
            \Return{$(s', A, m')$}\;
        }
\end{algorithm}

\subsection{Algorithm for proof of Lemma \ref{lemma.renyi-general-case}}

In this section, we enclose the algorithm for I.M.-to-I.M. postprocessing algorithm $\cP$ for the proof of Lemma \ref{lemma.renyi-general-case}, as defined in Algorithm \ref{alg.rdp_post}. This algorithm assumes inductively that $\calF_{K - 1}$-$FiltCon(IM)$ is a valid $\text{RDP}_\alpha$ filter. For any adversary, when receiving the adversary's request to start $\calM_1$ and verifying the privacy loss is within budget, $\cP$ will start both $\calM_1$ and $\calF_{K - 1}$-$FiltCon(IM)$ in $\calF_2$-$FiltCon(IM)$. $\cP$ will directly pass on queries to $\calM_1$. For queries $q$ to $\calM_j$ where $j > 1$ (in the form $(j, q)$), $\cP$ will postprocess the message as $(2, (j - 1), q)$, indicating that $\cP$ is passing $q$ on to the $(j - 1)$th mechanism of $\calF_{K - 1}$-$FiltCon(IM)$.

\begin{algorithm}
\small
    \DontPrintSemicolon
    \caption{Interactive postprocessing $\cP$ that transforms $\calF_2(\cdot; \eps)$-$FiltCon(IM)$ into $\calF_K(\cdot; \eps)$-$FiltCon(IM)$.}\label{alg.rdp_post}
    \Fn{\AlgPostprocessing{$s, v, m$}} {
         $k \gets$ number of mechanisms (started or attempting to start)\;
         \If(\tcp*[f]{initialize $\cP$}){$s = \lambda$}{
            $s \gets (k, [])$, where $[]$ is an empty list, $v \gets Q$\;
         }
         \If(\tcp*[f]{$\cP$ accepting queries}){$v = Q$}{
            \If(\tcp*[f]{$\cA$ starts a new mechanism}){$m$ is of the form $m = (\cM', (\alpha, \epsilon'))$}{
                \If(\tcp*[f]{currently no mechanisms}){$k = 0$}{
                    $\cM_1 \gets \cM'$, $\epsilon_1 \gets \epsilon'$\;
                    $b_1 \gets \texttt{no}$ \tcp*[f]{$b_1$ is the boolean keeping track of whether $\cM_1$ has successfully started}
                    $k \gets k + 1$, $s' \gets (k, [(\cM_1, (\alpha, \epsilon_1), b_1)])$, $m' \gets m$\;
                }\Else(\tcp*[f]{both $\calM_1$ and $\calF_{K - 1}$ started}){
                    $s' \gets s$, $m' \gets m$ \tcp*[f]{pass on query; the filter will handle any rejections (any request will be over budget)}
                }
            }\ElseIf(\tcp*[f]{$\cA$ queries an old mechanism}){$m$ is of the form $m = (j, q)$ where $j = 1, \dots, K$ and $q$ is a query to $\cM_j$}{
                \If(\tcp*[f]{querying $\cM_1$}){$j = 1$}{
                    $m' \gets (2, (j - 1, q))$, $s' \gets s, m' \gets m$ \;
                }
            }
            $v \gets A$ \tcp*[f]{we will be talking to the mechanism next}\;
            \Return{$(s', Q, m')$}\;
         }\ElseIf(\tcp*[f]{$\cP$ accepting answers}){$v = A$}{
            \If(\tcp*[f]{started $\calM_1$}){$k = 1$, $b_1 = \texttt{no}$, and $m$ is not an error message}{
                $b_1 \gets \texttt{yes}$, $b_2 \gets \texttt{no}$, $m' \gets (\calF_{K - 1}\text{-}FiltCon(IM), (\alpha, \epsilon - \epsilon_1))$\;
                $k \gets k + 1$, $m_\calM \gets m$\;
                $s' \gets (k, [(\calM_1, (\alpha, \epsilon_1), b_1), (\calF_{K - 1}\text{-}FiltCon(IM), (\alpha, \epsilon - \epsilon_1), b_2, m_\calM)])$ \tcc{start $\calF_{K - 1}$-$FiltCon(IM)$ immediately, store $m$ for return after $\calF_{K - 1}$ starts}
                \Return{$(s', Q, m')$}\;
            }\ElseIf(\tcp*[f]{just started $\calF_{K - 1}$}){$k = 1$ and $b_2 = \texttt{no}$}{
                Parse $s$ as $s = (2, [(\calM_1, (\alpha, \epsilon_1), b_1), (\calF_{K - 1}\text{-}FiltCon(IM), (\alpha, \epsilon - \epsilon_1), b_2, m_\calM)])$\;
                $b_2 \gets \texttt{yes}$\;
                $m' \gets m_\calM$ \tcp*[f]{pass on the answer from $\calM_1$}\;
                $v \gets Q$ \tcp*[f]{talking to adversary next}\;
                \Return{$(s', A, m')$}\;
            }\Else{
                 $s' \gets s, v \gets Q, m' \gets m$\;
                 \Return{$(s', A, m')$}\;
            }
         }
    }
\end{algorithm}

\subsection{Algorithm for proof of Theorem \ref{thm.zcdp-filter}}

In this section, we enclose the algorithm for I.M.-to-I.M. postprocessing algorithm $\cP$ for the proof of Theorem \ref{thm.zcdp-filter}, as defined in Algorithm \ref{alg.zcdp_post}. Essentially, $\cP$ only postprocesses queries that start new $\rho$-zCDP mechanisms into queries that start new $(\alpha, \rho\alpha)$-RDP mechanisms. All other interactions are preserved.

\begin{algorithm}
    \DontPrintSemicolon
    \caption{Interactive postprocessing $\cP$ that transforms $\calF^{(\alpha)}(\cdot; \eps)$-$FiltCon(IM)$ to $\calF(\cdot; \rho)$-$FiltCon(IM)$.}\label{alg.zcdp_post}
        \Fn{\AlgPostprocessing{$s, v, m$}}{
            \If(\tcp*[f]{initialize $\cP$}){$s = \lambda$}{
                $s \gets ([])$, where $[]$ is an empty list\;
                $v \gets Q$\;
            }

            \If(\tcp*[f]{$\cP$ accepting queries}){$v = Q$}{
                \If(\tcp*[f]{$\cA$ starts a new mechanism}){$m$ is of the form $m = (\calM', \rho')$}{
                    $m' \gets (\calM', (\alpha, \rho'\alpha))$\;
                }\ElseIf{$m$ is of the form $m = (j, q)$ where $j = 1, 2, \dots$ and $q$ is a query to $\calM_j$}{
                    $m' \gets m$ \tcp*[f]{ferry query to mechanism}\;
                }\ElseIf{$m$ cannot be parsed correctly}{
                    $m' \gets \texttt{invalid query}$\;
                }
                $v \gets A$ \tcp*[f]{talk to the mechanism next}\;
                \Return{$(s, Q, m')$}\;
            }\ElseIf(\tcp*[f]{$\cP$ accepting answers}){$v = A$}{
                $v \gets Q$ \tcp*[f]{talk to adversary next}\;
                \Return{$(s, A, m)$} \tcp*[f]{ferry answer back}\;
            }
            
        }
\end{algorithm}

\section{Proofs}
\subsection{Proof of Lemma \ref{lemma.f_odom}}\label{appendix.f_odom_proof}

In this section, we enclose the full proof of Lemma \ref{lemma.f_odom}.

\begin{lemma}[Lemma \ref{lemma.f_odom} restated] 
\begin{enumerate}
    \item A function $\calG:\calD^*\rightarrow \calD$ is a valid NIM-privacy-loss accumulator if and only if $\calF: \calD^* \times \calD \to \{0, 1\}$ constructed from $\calG(\cdot)$ such that $\calF(\cdot;d) = \mathbb I(\calG(\cdot) \preceq d)$ is a valid $\calD$-DP NIM-filter.
    \item A function $\calG:\calD^*\rightarrow \calD$ is a valid concurrent IM-privacy-loss accumulator if and only if $\calF: \calD^* \times \calD \to \{0, 1\}$ constructed from $\calG(\cdot)$ such that $\calF(\cdot;d) = \mathbb I(\calG(\cdot) \preceq d)$ is a valid $\calD$-DP concurrent IM-filter.
\end{enumerate}
\end{lemma}

\begin{proof}
    We will prove the bijection relation for IM, as the proof for NIM holds similarly. Fix a pair of datasets $x, x'$.
    We will prove Lemma \ref{lemma.f_odom} for the interactive case, because the non-interactive case follows similarly. 

    Fix a pair of datasets $x, x'$.

    \emph{If.} Suppose that $\calF(\cdot; d) = \mathbb I(\calG(\cdot) \preceq d)$ is a valid $\calD$ DP concurrent IM-filter. Then by Definition \ref{def.ddp_cf}, $\calF(\cdot; d)$-$FiltCon(IM)$ is a $d$-$\calD$ DP interactive mechanism for every $d\in \calD$. Thus, for an adversary $\calA$ interacting with $\CO$, we can construct an interactive mechanism $\cM$ that will behave exactly like $\CO$ as in Algorithm \ref{alg.co_construction}, but on every round of operation also checks if $\calG(\cdot) \not\preceq d$, in which case $\cM$ will maintain its previous state and output \texttt{halt} message. It is easy to see that
    $\View(\cA \leftrightarrow \cM(x))$ is exactly the same as $\texttt{Trunc}_d(\View(\calA \leftrightarrow \CO(x)))$ for any database $x$.

    We now prove that $\cM$ is an interactive postprocessing of $\calF(\cdot; d)$-$FiltCon(IM)$. Define a postprocessing $\cP$ as in Algorithm \ref{alg.bijection_left}. This mechanism allows $\cM$ to behave the same as $\calF(\cdot; d)$-$FiltCon(IM)$, just adding additional support for the deterministic \texttt{privacy\_loss} queries, which do not change the state of the odometer and therefore do not affect divergence.

    Therefore, 
       \begin{align}
       & D(\texttt{Trunc}_d(\View(\calA\leftrightarrow \CO(x)))|| \\ & \texttt{Trunc}_d(\View(\calA\leftrightarrow \CO(x'))))\notag \\
       & = D(\texttt{View}(\calA \leftrightarrow \cM(x))||\View(\calA \leftrightarrow \cM(x')))\notag\\
        & =D(\texttt{View}(\calA \leftrightarrow \cP \circ \CF(x)|| \\ & \texttt{View}(\calA \leftrightarrow \cP \circ \CF(x'))) \notag \\
        & \preceq D(\texttt{View}(\mathcal{A} \leftrightarrow \CF(x))|| \texttt{View}(\mathcal{A} \leftrightarrow \CF(x'))) \label{ineq.post1} \\
        & \preceq d,
        \end{align}
    where ~\eqref{ineq.post1} follows from Theorem \ref{thm.pim_privacy}.
    By Definition \ref{def.ddp_odom}, $\calG: \calD^* \to \calD$ is a valid $\calD$ DP privacy-loss accumulator.

    \emph{Only if.} Suppose that $\calG: \calD^* \to \calD$ is a valid $\calD$ DP privacy-loss accumulator. For every $d\in \calD$, because the continuation rule is $\calF(\cdot; d) = \mathbb I(\calG(\cdot) \preceq d)$,  $\calF(\cdot; d)$-$FiltCon(IM)$ will halt when $\calG(\cdot) \not\preceq d$, which is the same condition for outputting $\View$ truncated at $d$ as in Algorithm \ref{alg:truncview}. We can use the same $\cM$ as in the previous direction that is constructed such that, for an adversary $\cA$ and database $x$, $\View(\cA \leftrightarrow \cM(x))$ is exactly the same as $\texttt{Trunc}_d(\View(\calA \leftrightarrow \CO(x)))$.

        We now define an interactive postprocessing $\cP$ that transforms $\cM$ into $\calF(\cdot; d)$-$FiltCon(IM)$ in
        Algorithm \ref{alg.bijection_right}, which removes support for $\texttt{privacy\_loss}$ queries since $\calF(\cdot; d)$-$FiltCon(IM)$ would not support that. Other interaction is essentially preserved, since $\cM$ and $\calF(\cdot; d)$-$FiltCon(IM)$ will both halt when $\calG(\cdot) \not\preceq d$. 
        
        Again, because \texttt{privacy\_loss} queries do not affect divergence, we have
        \begin{align}
            & D(\texttt{View}(\calA \leftrightarrow \CF(x) ||  \texttt{View}(\calA \leftrightarrow \CF(x'))) \notag \\
            & = D(\texttt{View}(\calA \leftrightarrow \cP \circ \cM(x) || \View(\calA \leftrightarrow \cP \circ \cM(x')))  \notag \\
            & \preceq D(\View(\calA \leftrightarrow \calM(x))||\View(\calA \leftrightarrow \calM(x'))) \label{ineq.post2} \\
        & = D( \texttt{Trunc}_d(\texttt{View}(\mathcal{A} \leftrightarrow \CO)(x) || \\ & \texttt{Trunc}_d(\texttt{View}(\mathcal{A}\leftrightarrow \CO)(x'))) \notag \\
            & \preceq d 
        \end{align}
       where ~\eqref{ineq.post2} follows from Theorem \ref{thm.pim_privacy}.
        Thus, $\calF(\cdot; d)$-$FiltCon(IM)$ is a $d$-$\calD$ DP interactive mechanism. By Definition \ref{def.ddp_cf}, since this is true for arbitrary $d$, $\calF$ is a valid $\calD$ DP filter.
    \end{proof}

\begin{algorithm}
    \DontPrintSemicolon
    \caption{Interactive postprocessing $\cP$ that transforms $\calF(\cdot; d)$-$FiltCon(IM)$ into $\cM$}\label{alg.bijection_left}
        \Fn{\AlgPostprocessing{$s, v, m$}}{
            \If(\tcp*[f]{initialize $\cP$}){$m = \lambda$}{
                $s' \gets ([])$, where $[]$ is an empty list\;
            }

            \If(\tcp*[f]{accepting queries from adversary}){$v = Q$}{
                \If{$m$ is of form $m = (\calM', d')$}{
                    Parse $s$ as $([\calM_1, \dots, \calM_k])$\;
                    $s' \gets ([\calM_1, \dots, \calM_k, (\calM_{k + 1}, \texttt{pending})])$\;
                }\ElseIf{$m$ is of form $m = (j, q)$ where $j = 1, \dots, k$ and $q$ is a query to $\calM_j$}{
                    $s' \gets s$\;
                }\ElseIf{$m = \texttt{privacy\_loss}$}{
                    $m' \gets \calG(\cdot)$ \;
                    \Return{$(s', A, m')$}\;
                }
                $m' \gets m$\;
                $v' \gets Q$, $v \gets A$ \tcp*[f]{pass query on to mechanism}\;
            }\ElseIf(\tcp*[f]{accepting answers from filter}){$v = A$}{
                \If{$s$ is of form $s = ([\calM_1, \dots, (\calM_k, \texttt{pending})])$}{
                    \If(\tcp*[f]{successfully started new mechanism}){$m$ is not an error message}{
                        $s' \gets ([\calM_1, \dots, \calM_k])$\;
                    }\Else{
                        $s' \gets ([\calM_1, \dots, \calM_{k - 1}])$\;
                    }
                }\Else{
                    $s' \gets s$\;
                }
                $m' \gets m$, $v' \gets A$, $v \gets Q$ \tcp*[f]{pass query on to adversary}\;
            }
            \Return{$(s', v', m')$}\;
        }
\end{algorithm}

\begin{algorithm}
    \DontPrintSemicolon
    \caption{Interactive postprocessing $\cP$ that transforms $\cM$ into $\calF(\cdot; d)$-$FiltCon(IM)$}\label{alg.bijection_right}
        \Fn{\AlgPostprocessing{$s, v, m$}}{
            \If(\tcp*[f]{initialize $\cP$}){$s = \lambda$}{
                $s \gets ([])$, where $[]$ is an empty list\;
            }
            \If(\tcp*[f]{accepting queries from adversary}){$v = Q$}{
                \If{$m$ is of form $m = (\calM', d')$}{
                    Parse $s$ as $([\calM_1, \dots, \calM_k])$\;
                    $s' \gets ([\calM_1, \dots, \calM_k, (\calM_{k + 1}, \texttt{pending})])$\;
                }\ElseIf{$m$ is of form $m = (j, q)$ where $j = 1, \dots, k$ and $q$ is a query to $\calM_j$}{
                    $s' \gets s$\;
                }\ElseIf(\tcp*[f]{remove support for privacy loss queries}){$m$ is of form $m = \texttt{privacy\_loss}$}{
                    $m' \gets \texttt{invalid query}$ \;
                    \Return{$(s', A, m')$}\;
                }

                $m' \gets m$\;
                $v' \gets Q$, $v \gets A$ \tcp*[f]{pass query on to mechanism}\;
            }\ElseIf(\tcp*[f]{accepting answers from filter}){$v = A$}{
                \If{$s$ is of form $s = ([\calM_1, \dots, (\calM_k, \texttt{pending})])$}{
                    \If(\tcp*[f]{successfully started new mechanism}){$m$ is not an error message}{
                        $s' \gets ([\calM_1, \dots, \calM_k])$\;
                    }\Else{
                        $s' \gets ([\calM_1, \dots, \calM_{k - 1}])$\;
                    }
                }\Else{
                    $s' \gets s$\;
                }
                $m' \gets m$, $v' \gets A$, $v \gets Q$ \tcp*[f]{pass query on to adversary}\;
            }
            \Return{$(s', v', m')$}\;
        }
\end{algorithm}

\subsection{Proof of Example \ref{example.eps-delta-filter} ($(\eps, \delta)$-DP filter)}

In this section, we demonstrate how the $(\eps, \delta)$-DP filter and odometer for noninteractive mechanisms defined in \cite{whitehouse2023fully} correspond to the filter and odometer we construct in Algorithm \ref{alg.k_noninteractive_comp} and Algorithm \ref{alg.o_construction}, respectively.

\begin{example}[Example \ref{example.eps-delta-filter} restated]
    For every $\delta' > 0$, 
    \begin{align*}
        & \calF((\eps_1, \delta_1), \dots, (\eps_k, \delta_k); (\eps, \delta)) \\ 
        & = \mathbb I\left(\epsilon \leq \sqrt{2\log\left(\frac{1}{\delta'}\right)\sum_{m \leq k}\eps_m^2} + \frac{1}{2}\sum_{m \leq k} \eps_m^2\right) \cdot \mathbb I\left(\delta' + \sum_{m \leq k} \leq \delta\right)
    \end{align*}
    is a valid $(\eps, \delta)$-DP privacy filter. 
\end{example}

\begin{proof}
    This filter is translated from the following theorem from \cite{whitehouse2023fully}, who define an  $(\eps, \delta)$-DP privacy filter in their notation as follows: 

    \begin{theorem}[$(\eps, \delta)$-DP privacy filter \cite{whitehouse2023fully}]
        \label{example.wrrw-eps-delta-filter}
        Suppose $\calM_1, \calM_2, \dots$ is a sequence of mechanisms such that, for any $k \geq 1$, $\calM_k$ is $(\epsilon_k, \delta_k)$-DP conditioned on $\calM_1, \dots, \calM_{k - 1}$. ($\calM_1, \calM_2, \dots$ are random variables that take on values that are the possible outputs of the respective mechanism on a dataset.) Let $\eps > 0$ and $\delta = \delta' + \delta''$ be a target approximation parameter such that $\delta' > 0$, $\delta'' \geq 0$. Define $N := N((\eps_1, \delta_1), (\eps_2, \delta_2), \dots) := \inf\{ k \in \mathbb N: \eps < \sqrt{2\log\left(\frac{1}{\delta'}\right)\sum_{m \leq k + 1}\eps_m^2} + \frac{1}{2}\sum_{m \leq k + 1} \eps_m^2 \text{ or } \delta'' < \sum_{m \leq k + 1} \delta_m \}$. Then the composition of these mechanisms $Comp(\calM_1, \dots, \calM_{N(\cdot)})$ is $(\eps, \delta)$-DP. 
    \end{theorem}
    $\calM_k$ being $(\eps_k, \delta_k)$-DP conditioned on the previous mechanisms means the privacy-loss parameters are chosen adaptively. This is analogous to our setting for filters as defined in Algorithm \ref{alg.k_noninteractive_comp}, in which the adversary can send in a query that decomposes into a mechanism and distance, a signal to the filter to start a new mechanism. Since the query could depend on the history of interaction, the new mechanism's parameters are adaptively chosen.
    
    We define $\calF: \calD^* \to \{ 1, 0\}$ as 
    \begin{align*}
        & \calF((\eps_1, \delta_1), \dots, (\eps_k, \delta_k); (\eps, \delta)) \\ 
        & = \mathbb I\left(\sqrt{2\log\left(\frac{1}{\delta'}\right)\sum_{m \leq k}\eps_m^2} + \frac{1}{2}\sum_{m \leq k} \eps_m^2 \leq \eps \right) \cdot \mathbb I\left(\delta' + \sum_{m \leq k} \delta_m \leq \delta \right),
    \end{align*}
    meaning that a mechanism equipped with $\calF$ will continue until for some $k$, $\sqrt{2\log\left(\frac{1}{\delta'}\right)\sum_{m \leq k + 1}\eps_m^2} + \frac{1}{2}\sum_{m \leq k + 1} \eps_m^2 > \epsilon$ or $\sum_{m \leq k + 1} \delta_m > \delta''$. 
    This is equivalent to \cite{whitehouse2023fully}'s notion of $N(\cdot)$. Both the condition on $\epsilon$ and the condition on $\delta''$ are increasing in the number of mechanisms being composed. The difference is that in Whitehouse et al.'s \cite{whitehouse2023fully} setting, the number of mechanisms being composed $N(\cdot)$ is the smallest number for adding another mechanism will cause $\sqrt{2\log\left(\frac{1}{\delta'}\right)\sum_{m \leq N(\cdot) + 1}\eps_m^2} + \frac{1}{2}\sum_{m \leq N(\cdot) + 1} \eps_m^2$ to exceed $\epsilon$, or in other words, $N(\cdot)$ is the largest number for which $\sqrt{2\log\left(\frac{1}{\delta'}\right)\sum_{m \leq N(\cdot)}\eps_m^2} + \frac{1}{2}\sum_{m \leq N(\cdot)} \eps_m^2 \leq \epsilon$. Since we define our filter in terms of a bound on the privacy parameter instead of a bound on the number of mechanisms, this is equivalent to checking that for $k$ mechanisms being composed, $\sqrt{2\log\left(\frac{1}{\delta'}\right)\sum_{m \leq k}\eps_m^2} + \frac{1}{2}\sum_{m \leq k} \eps_m^2 \leq \epsilon$. The same holds for the $\sum_{m \leq k} \delta_m \leq \delta''$ term. Since $\delta'$ is already set, we fold the $\sum_{m \leq k}\delta_m \leq \delta''$ requirement into the $\delta' + \delta'' \leq \delta$ requirement to get the condition $\delta' + \sum_{m \leq k}\delta_m \leq \delta$. Taking the lower bound of the $k$ that satisfies either $\sqrt{2\log\left(\frac{1}{\delta'}\right)\sum_{m \leq k}\eps_m^2} + \frac{1}{2}\sum_{m \leq k} \eps_m^2 \leq \epsilon$ or $\delta' + \sum_{m \leq k}\delta_m \leq \delta$ is equivalent to finding the $k$ that satisfies both conditions. 

    With fully adaptive composition of noninteractive mechanisms, the mechanism will only halt if starting a new mechanism would exceed its privacy-loss budget. Since $N(\cdot)$ mechanisms composed are $(\eps, \delta)$-DP in \cite{whitehouse2023fully}'s setting, we know that a mechanism equipped with $\calF$ that has not yet halted will be $(\eps, \delta)$-DP as well.   

    Then by Definition \ref{def.ddp_f}, $\NIF$ is an $(\eps, \delta)$-DP interactive mechanism, meaning that $\calF$ is a valid $(\eps, \delta)$-DP for the sequence of $(\eps_i, \delta_i)$-DP noninteractive mechanisms as defined above. 
\end{proof}

To define a concurrent privacy odometer, we make use of Lemma \ref{lemma.filter_to_odo}.

\begin{theorem}[$(\eps, \delta)$-DP privacy odometer]
    Let $\delta = \delta' + \delta''$ be a target approximation parameter such that $\delta' > 0$, $\delta'' \geq 0$. Then 
    \begin{align*}
        & \calG((\eps_1, \delta_1), \dots, (\eps_k, \delta_k)) \\
        & = \begin{cases}
            \left(\sqrt{2\log\left(\frac{1}{\delta'}\right)\sum_{m \leq k}\eps_m^2} + \frac{1}{2}\sum_{m \leq k} \eps_m^2, \delta\right) & \delta' + \sum_{m \leq k} \delta_m \leq \delta \\
            (\infty, \infty) & \text{otherwise}
        \end{cases}
    \end{align*}
    
    is a valid $(\eps, \delta)$-DP privacy-loss accumulator for noninteractive mechanisms, where $\eps = \sqrt{2\log\left(\frac{1}{\delta'}\right)\sum_{m \leq k}\eps_m^2} + \frac{1}{2}\sum_{m \leq k} \eps_m^2$.
\end{theorem}

\end{document}